\documentclass{article}

\usepackage{enumerate}
\usepackage{amsmath}
\usepackage{amssymb}
\usepackage{amsthm}
\usepackage{caption}
\usepackage{geometry}
\usepackage{graphicx}
\usepackage{color}
\usepackage{url}
\usepackage{textcomp}
\usepackage{hyperref}
\usepackage{authblk}
\usepackage{subfigure}
\usepackage{setspace}
\usepackage{listings}
\usepackage{booktabs}
\usepackage{multirow}
\usepackage[bottom]{footmisc}
\usepackage{indentfirst}
\usepackage[numbers,sort&compress]{natbib}
\usepackage{tikz}
\usetikzlibrary{arrows.meta, positioning}

\newcommand{\etal}{\textit{et al}.\ }
\newcommand{\per}{\mathit{per}}

\DeclareCaptionFormat{custom-lst}{\hrule#1#2#3}
\hypersetup{colorlinks=true, linkcolor=blue, citecolor=blue}

\newtheorem{theorem}{Theorem}
\newtheorem{lemma}{Lemma}
\newtheorem{proposition}{Proposition}
\theoremstyle{remark}
\newtheorem{remark}{Remark}
\newtheorem{assumption}{Assumption}

\title{\large \bf Energy-minimized interpolatory B-spline lofting for sequences of closed contours}
\date{}
\author[a,*]{Shutao Tang}
\author[b]{Zixuan Zhao}
\affil[a]{\small Department of Mechanical Engineering, Northwestern Polytechnical University, Xi'an, China}
\affil[b]{\small The University of Chicago, 5801 S Ellis Ave, Chicago, IL 60637, United States}

\begin{document}

\maketitle

\begin{abstract}
    \noindent B-spline lofting through closed cross-sections is common in reverse engineering, 
    but interpolating each section on its own knot vector can produce an unnecessarily large control net.
    We construct a common-knot method for closed section curves of odd and even degree by minimizing the strain energy of each section subject to its interpolation constraints.
    Feasibility for closed section data and uniqueness of the resulting saddle-point solution require full row rank of the periodic collocation matrices, 
    a property that a previously proposed periodic Schoenberg--Whitney window condition does not always guarantee.
    We establish the exact parity distinction: the condition is sufficient at odd point counts but admits singular configurations at every admissible even point count.
    For the odd and even degrees used here, the natural and shifting anchor knots are nonsingular; 
    at odd point counts the window condition certifies the selected candidate, whereas at even point counts a rank test with an anchor fallback is applied.
    The lofting method therefore retains a nonsingular square subsystem for each section and uses nestedness to transfer full row rank to the rectangular common-knot system.
    Across six closed-contour models the common-knot method uses 83--98\% fewer control points than traditional interpolating lofting, 
    with section energies no larger than those of a previous knot-selection method within numerical precision.

    \vspace{1.5mm}
\end{abstract}

\section{Introduction}
\label{sec:introduction}

Surface skinning or lofting~\cite{woodward1987cross,woodward1988skinning} constructs a smooth surface through a sequence of sectional curves.
In reverse engineering, these sections often come from slices of a scanned object and are closed contours when the object is tubular.
When the measured contours are clean or are treated as hard constraints, 
an interpolatory B-spline representation is desirable because it preserves the measured contours and 
produces a tensor-product surface suitable for downstream CAD operations~\cite{piegl1996nurbs}.
For interpolatory lofting, the difficulty is the size of the resulting control net.
If each section is interpolated on its own knot vector, the tensor-product surface must carry the union of all section knots, 
so the number of control points can grow far beyond the number needed by any single section~\cite{piegl2000reducing,park2003lofted,wang2008reducing,park2004new}.

For interpolating open sections, a common remedy is to interpolate all sections on one reduced knot vector.
The section interpolation systems then become under-determined, and a strain-energy minimizer selects one interpolant on the common knots~\cite{park2003lofted}.
Full row rank of the section collocation matrices guarantees feasibility of the interpolation and nonsingularity of the saddle-point system.
Wang~\etal\cite{wang2008reducing} identified the corresponding full-rank condition for open curves.
Closed sections require a periodic collocation matrix, and the open-curve Schoenberg--Whitney argument does not transfer directly.
A previous closed-contour method~\cite{tang2023control} introduced a periodic window condition and used it as the solvability test for knot selection.
The energy-minimizing strategy~\cite{park2003lofted}, however, has not been extended to closed contours; the present paper fills this gap.

Approximation lofting is the standard fitting-based route when the input is a raw scan: each section is fitted within a tolerance rather than forced through every measured point, 
which avoids reproducing measurement noise and can yield a smaller control net~\cite{piegl1996algorithm,park1996smooth,piegl2000surface,park2000method,piegl2002surface}.
Interpolation remains preferable when the sections are exact CAD curves or clean, consistently sampled contours; 
when individual measurements are hard constraints, such as gauge or certification points; 
when downstream CAD/CAE operations require the surface to pass through the given sections; or when small geometric features must not be smoothed away.
It also avoids the need to choose a fitting tolerance and to correct parameters, and it is natural for sparse sections.
The two branches share the compatibility and control-point-reduction problem, but differ in whether the measured points are hard constraints.
This paper therefore addresses the interpolation branch; 
for raw noisy scans, approximation lofting is the appropriate complementary tool, and the method below can be applied after denoising and segmentation, or when exact section fidelity is mandated.

This paper establishes the full-rank conditions needed for energy-minimizing lofting of closed contours and develops the corresponding lofting algorithm on that basis.
The previously proposed window condition has a sharp parity distinction: it is sufficient when the number of interpolation points is odd (Theorem~\ref{thm:odd-window}), 
but for every degree it admits singular configurations at admissible even point counts (Proposition~\ref{prop:singular}).
Because at even point counts the window condition may fail, the method applies a direct rank check to the selected knots of such sections, 
with the anchor knots as a fallback; the anchor configurations produce a nonsingular collocation matrix on given parameters.
Nestedness then transfers full row rank of the resulting square subsystem to the rectangular common-knot system (Theorem~\ref{thm:full-rank-lift}).
The main contributions of this work can be summarized as follows:

  \begin{itemize}
      \item \textbf{Deficiency of the periodic window condition.}
            We identify the parity defect of the periodic window condition used for closed-contour knot selection: it is sufficient at odd point counts (Theorem~\ref{thm:odd-window}), but the uniform and non-uniform counterexamples of Proposition~\ref{prop:singular} show that it fails to certify solvability at even point counts.

      \item \textbf{Periodic full-row-rank lifting theorem.}
            We prove that a nonsingular square subsystem, together with nestedness of the periodic spline spaces, certifies full row rank of the rectangular common-knot system (Theorem~\ref{thm:full-rank-lift}).

      \item \textbf{Energy-minimizing lofting for closed contours.}
            We extend the energy-minimizing lofting algorithm to closed contours, combining safe anchor configurations with a parity-aware rank check to build a common-knot method that applies to both odd and even degrees.
  \end{itemize}

The framework of this paper is organized in Figure~\ref{fig:framework}.

\begin{figure}[htbp]
    \centering
    \begin{tikzpicture}[
        >=Stealth,
        font=\small,
        node distance=7mm,
        box/.style={draw, rounded corners=2pt, align=center, inner sep=4pt, text width=6.4cm},
        tech/.style={draw, rounded corners=2pt, align=center, inner sep=4pt, text width=4.4cm},
        outbox/.style={draw, rounded corners=2pt, align=center, inner sep=6pt, text width=12cm, fill=gray!10}
    ]
        \node[box] (interp) {Closed B-spline interpolation};

        \node[tech, below=of interp, xshift=-5.2cm] (t1) {anchor knots\\ (\emph{natural}/\emph{shifting}~\cite{park2001choosing})\\[2pt] $\Rightarrow$ nonsingular collocation matrix};
        \node[tech, below=of interp] (t2) {periodic window condition~\cite{tang2023control}\\[2pt] \emph{odd} \# control points: window holds (Theorem~\ref{thm:odd-window})\\ \emph{even} \# control points: window may fail,\\ rank check detects this with some probability};
        \node[tech, below=of interp, xshift=5.2cm] (t3) {full row rank lifting (Theorem~\ref{thm:full-rank-lift})\\ transfer full row rank};

        \draw[->] (interp.south) -- (t1.north);
        \draw[->] (interp.south) -- (t2.north);
        \draw[->] (interp.south) -- (t3.north);

        \node[outbox, below=of t2] (alg) {Energy-minimizing lofting algorithm for closed contours};
        \draw[->] (t1.south) -- (t1.south |- alg.north);
        \draw[->] (t2.south) -- (alg.north);
        \draw[->] (t3.south) -- (t3.south |- alg.north);
    \end{tikzpicture}
    \caption{The framework of this paper: the \emph{natural} and \emph{shifting} methods yield nonsingular anchor knots, which, together with the periodic window condition and the full-row-rank lifting theorem, lead to the energy-minimizing lofting algorithm for closed contours.}
    \label{fig:framework}
\end{figure}
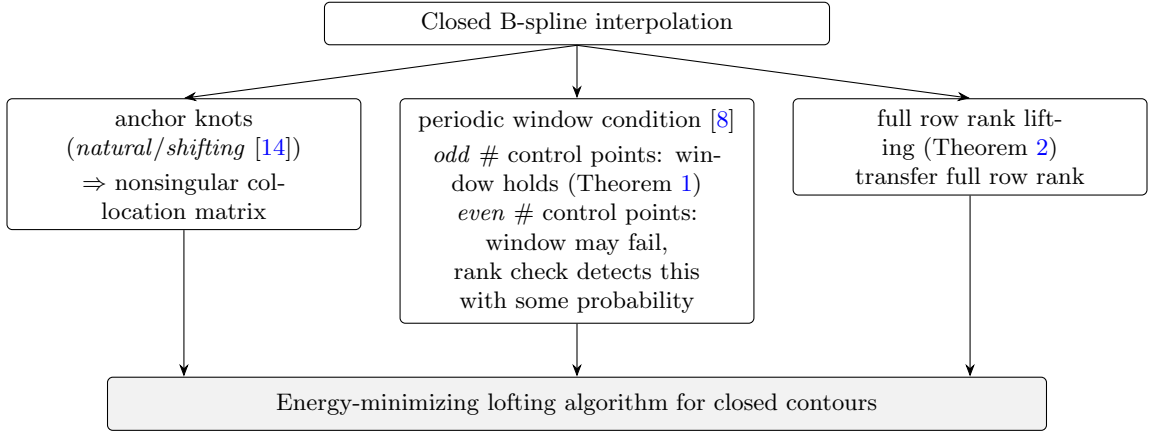

The remainder of this paper is organized as follows. Section~\ref{sec:related}
reviews related work on B-spline lofting and the solvability of periodic interpolation.
Section~\ref{sec:background} states the lofting problem and introduces periodic B-splines
on the parameter circle. Section~\ref{sec:theory} analyzes the solvability of the periodic
collocation matrix, identifies the deficiency of the window condition, and establishes
the full-row-rank lifting theorem. Section~\ref{sec:certified} develops the energy-minimizing
common-knot lofting algorithm with its anchor configurations and rank checks.
Section~\ref{sec:experiments} reports experimental results on six closed-contour models.
Section~\ref{sec:discussion} discusses even-cardinality behavior and the anchor fallback, the window-factor trade-off, and the relation to periodic T-spline skinning.
Section~\ref{sec:conclusions} concludes the paper and outlines future work.
Proofs of the longer results are collected in Appendix~\ref{app:proofs}.

\section{Related work}
\label{sec:related}
This section reviews approximation-based and interpolation-based B-spline lofting, with emphasis on control-point reduction and on the full-rank conditions that govern the common-knot interpolation systems. 
It also connects classical solvability results for periodic spline interpolation.

\subsection{Approximation lofting}
Approximation lofting fits each section within a tolerance instead of forcing it through every measured point.
It is the standard fitting-based route for noisy, oversampled, or inconsistently sampled scans, where exact interpolation would reproduce measurement noise; the price is that fine features may be smoothed and the result depends on the chosen tolerance.
Because approximation methods do not impose interpolation equations, the periodic collocation rank problem studied here does not arise in the same form, although rank and regularization still affect the least-squares systems.
Approximation is therefore complementary to the present work: it is preferable for raw noisy scans, whereas interpolation is preferable when the section data are exact or must be honored exactly.

\subsubsection{Lofting from section curves}
Surface lofting (skinning) of section curves is classical~\cite{piegl1996nurbs}.
Making the section curves compatible merges their knot vectors, and this union grows with the number of curves, so the compatible representation can become prohibitively large.
Approximate skinning first makes the curves compatible by exact knot insertion and then progressively removes knots, a memory-intensive approach~\cite{piegl1996algorithm}.
Smooth surface approximation to serial cross-sections~\cite{park1996smooth} addresses the same compatibility problem, and knot removal~\cite{lyche1987knot} remains the standard tool for shrinking a merged vector afterwards.
Park~\etal\cite{park2000method} instead make the curves compatible by refitting every curve on a single common knot vector: they sample each input curve, parameterize the samples by normalized arc length, choose a common control count, build a local knot vector for each curve reflecting its parameter distribution, average the local vectors element-wise into the common one, and fit every curve on it with stretching and bending energies; their method also covers closed curves on cyclic knot vectors.
Piegl and Tiller~\cite{piegl2002surface} proposed another curve-sampling algorithm similar to the common-knot refitting of Park~\etal\cite{park2000method}, sampling the input curves and then fitting the sampled data points.

\subsubsection{Fitting rows of data points}
When the input consists of rows of measured points rather than fitted curves, each row is first fitted by a curve and the rows are then made compatible as above.
For approximation, Piegl and Tiller~\cite{piegl2000surface} fit scanned open rows of unequal size with an error-bounded scheme that reuses knots between rows and balances parameter counts among spans, 
so rows of similar shape share nearly all of their knots.
The common-knot refitting of Park~\etal\cite{park2000method} also applies to rows of data points and is directly relevant to the measured closed-row setting considered here.
The row-structured data studied here are not gridded in general, but when every row carries the same number of points they reduce to gridded (mesh-structured) data; 
tensor-product B-spline surfaces are the natural representation for both.
Efficient least-squares fitting to gridded data exploits the tensor-product structure for fast matrix assembly~\cite{merchel2022fast}, 
and low-rank spline approximation of gridded data has been studied in~\cite{mokris2024lowrank}.
Adaptive knot placement for gridded data was addressed by~\cite{park2011dominant}, where dominant columns of the data matrix determine knot locations for B-spline surface fitting.
Few of these works address the closed (periodic) setting considered here.

\subsection{Interpolating lofting and control-point reduction}
Interpolating lofting forces the surface through the section data; it is the method of choice in the scenarios listed in Section~\ref{sec:introduction}.
Early control-point reduction methods for lofted B-spline surface interpolation focused mainly on open sections.
Piegl and Tiller~\cite{piegl2000reducing} select section knots from a shared pool under the Schoenberg--Whitney condition and then use knot refinement to put the sections in a common space.
Park~\cite{park2003lofted} takes a different route: a reduced common knot vector is fixed first, and the under-determined section systems are solved by constrained energy minimization.
Wang~\etal\cite{wang2008reducing} extend the Schoenberg--Whitney condition to the under-determined common-knot systems (the SW condition extension); 
observing that these systems can lose rank, they give the corresponding repair for open curves.
Park~\etal\cite{park2004new} combine lofting with universal parametrization~\cite{lim1999universal,jung2000new}.

Closed contours introduce the additional issue that the section interpolation problem is periodic.
Interpolatory control-point reduction for closed contours was introduced by Tang~\cite{tang2023control},
who extended Piegl and Tiller's knot-reusing strategy to closed contours by replacing the open-curve
Schoenberg--Whitney test with a periodic window condition.
Table~\ref{tab:interp-summary} summarizes the control-point reduction techniques for rows of scanned data interpolation together with the solvability condition each one relies on.
For open sections, the two strategies rely on the Schoenberg--Whitney condition (SW condition for short) and its extension to under-determined common-knot systems. 
For closed sections, the energy-minimizing construction of this paper needs two key techniques: the periodic window condition, which certifies the square knot-selection subsystem, 
and the full-row-rank lifting theorem, which transfers full row rank to the rectangular common-knot system.

\begin{table}[htbp]
    \centering
    \caption{Control-point reduction techniques for rows of scanned data interpolation and the solvability condition each relies on.}
    \label{tab:interp-summary}
    \begin{tabular}{llll}
        \toprule
        Category & Strategy & Reference & Key technique\\
        \midrule
        \multirow{2}{*}{Open-section}
            & knot reusing & Piegl and Tiller~\cite{piegl2000reducing} & SW condition\\
            & energy minimization & Park~\cite{park2003lofted}, Wang~\etal\cite{wang2008reducing} & SW condition extension\\
        \multirow{2}{*}{Closed-contour}
            & knot reusing & Tang~\cite{tang2023control} & periodic window condition\\
            & energy minimization & this paper & full-row-rank lifting (Theorem~\ref{thm:full-rank-lift})\\
        \bottomrule
    \end{tabular}
\end{table}

T-spline skinning~\cite{nasri2012local,yang2012approximate,li2015surface,oh2018local,hu2020curve} addresses local refinement and non-tensor-product topology; 
these directions are complementary to the common-knot tensor-product setting considered here.
Among them, the periodic T-spline skinning of Li~\etal\cite{li2015surface} is closest to our setting, 
but its input is a family of closed NURBS section curves, whereas we start from rows of closed contour points; 
this difference is discussed in Section~\ref{sec:discussion}.

\subsection{Solvability of periodic spline interpolation}
Periodic spline interpolation differs structurally according to the dimension of the spline space.
Lee~\etal\cite{lee1983periodic} studied periodic B-spline collocation matrices in connection with infinite spline interpolation.
Odd-dimensional periodic spline spaces are weak Chebyshev, 
and their interpolation sets are characterized by a periodic Schoenberg--Whitney condition~\cite{schumaker2007spline,zeilfelder1996periodic,davydov2000weak}.
Golomb~\cite{golomb1968periodic} treated odd-degree interpolation on uniform partitions, 
while periodic odd-degree interpolation at arbitrary knots has the classical variational interpretation that underlies the natural method (Section~\ref{subsec:sufficient}).
Dubeau and Savoie~\cite{dubeau1985periodic} proved that even-degree periodic interpolation at the knots of a uniform partition is uniquely solvable exactly when the number of points is odd.
This parity effect is central to the counterexamples in Section~\ref{subsec:insufficiency}.
Meir and Sharma~\cite{meir1968convergence} asserted unisolvence for periodic cubics collocated at a fixed ratio of each knot interval; 
the counterexamples below show that this assertion fails for every even point count.
For closed B-spline curves, Park and Kim~\cite{park1996smooth} adapted knot averaging and observed its reliability for odd degrees, 
and Park~\cite{park2001choosing} introduced the \emph{natural}, \emph{shifting}, and \emph{energy} methods, where the \emph{natural} method is for odd degree and the latter two for even degree.
The natural and shifting configurations coincide with the anchor configurations analyzed in Section~\ref{subsec:sufficient}.

\section{Problem statement and preliminaries}
\label{sec:background}

For the reader's convenience, Table~\ref{tab:notation} collects the principal symbols used throughout the paper.

\subsection{The lofting problem}
\label{subsec:problem}

Given a sequence of closed contours $\mathcal{Q} = \{{\bf{Q}}_{i,j} \mid i = 0, \ldots, m;\; j = 0, \ldots, n_i\}$, 
aligned under Assumption~\ref{ass:aligned} below, we seek a closed-in-$v$ lofted B-spline surface
\begin{equation}
    \label{eq:surface}
    {\bf{S}}(u, v) \;=\; \sum_{k} \sum_{l} {\bf{P}}_{k,l}\, N_{k,p}(u)\, N^{\mathrm{per}}_{l,q,{\bf{D}}}(v)
\end{equation}
of longitudinal degree $p$ and section degree $q\ge2$ (odd or even) that interpolates all points with few control points.
Section $i$ is interpolated as a periodic B-spline curve on \emph{common domain knots} ${\bf{D}}$ shared by all sections; 
the control points of the section curves are then interpolated longitudinally (Algorithm~A9.4 of~\cite{piegl1996nurbs}).
We assume $n_i \ge q$, $m \ge p$, and distinct longitudinal parameters.
Within each closed contour consecutive points, including the pair across the seam, are distinct; 
a duplicated closing point ${\bf Q}_{i,n_i}={\bf Q}_{i,0}$ is removed before parametrization.

\begin{assumption}\label{ass:aligned}
The rows are aligned: the first point of each row is chosen so that
the seam points of consecutive rows correspond (``the line of first
points is not twisted'').  This is the standing assumption of row-wise
lofting; seam optimization is outside the scope of this paper.
\end{assumption}

For one section with points ${\bf{Q}}_0, \ldots, {\bf{Q}}_n$ (the section index is dropped when a single section is discussed), 
the parameters are the chord-length parameters of the closed contour,
\begin{equation}
    \label{eq:params}
    {\bf{T}} = \{t_j\}_{j=0}^{n}:
    \qquad
    t_0 = 0,
    \qquad
    t_j = t_{j-1} + \frac{\| {\bf{Q}}_j - {\bf{Q}}_{j-1} \|}{L},
    \qquad
    L = \sum_{j=1}^{n} \| {\bf{Q}}_j - {\bf{Q}}_{j-1} \| + \| {\bf{Q}}_0 - {\bf{Q}}_n \|,
\end{equation}
so that $t_0 < t_1 < \cdots < t_n < 1$ and the \emph{wrap gap} $d_n = 1 - t_n > 0$ closes the contour.
The parametrization is seam-normalized, $t_0 = u_0 = 0$.

Interpolating the largest section requires at least $\max_i n_i + 1$ periodic basis functions.
Thus every common-knot method satisfies the lower bound
\begin{equation}
    \label{eq:floor}
    \text{control points per section} \;=\; \hat n + q + 1 \;\ge\; \max_i(n_i) + q + 1,
\end{equation}
where $\hat n + 2$ is the size of ${\bf{D}}$ and the count includes the $q$ wrap-around control points of the standard closed representation~\cite{park1996smooth,tang2023control}.

\begin{table}[htbp]
    \centering
    \caption{Notation and definitions.}
    \label{tab:notation}
    \small
    \begin{tabular}{ll}
        \toprule
        Symbol & Meaning\\
        \midrule
        ${\bf{Q}}_{i,j}$ & data points; $m+1$ sections, $n_i+1$ points in section $i$\\
        $p$, $q$ & longitudinal degree and section degree; here $q\ge2$ (odd or even degree)\\
        ${\bf{D}} = \{u_0, \ldots, u_{\hat n + 1}\}$ & domain knots, $0 = u_0 < \cdots < u_{\hat{n}+1} = 1$, extended periodically\\
        ${\bf{T}} = \{t_j\}_{j=0}^{n}$, $\tilde t_j$, $d_n$ & chord-length parameters~\eqref{eq:params}, shifted parameters (even $q$), wrap gap\\
        ${\boldsymbol\tau} = \{\tau_j\}_{j=0}^{n}$ & collocation parameters of a section ($t_j$ for odd $q$, $\tilde t_j$ for even $q$)\\
        $N^{\mathrm{per}}_{i,q,{\bf{D}}}$ & periodic B-spline basis function centered at $u_i$, see~\eqref{eq:per-basis}\\
        ${\bf{C}}_{\hat n}$ & section curve: periodic B-spline on ${\bf{D}}$ with control-point matrix ${\bf{P}}$\\
        ${\bf{N}}({\bf{D}}, {\bf{T}})$ (or ${\bf{N}}$) & periodic collocation matrix~\eqref{eq:collocation}\\
        $\bar{{\bf{D}}}$ & selected domain knots (Procedure~\ref{lst:proc1}) whose collocation matrix is nonsingular\\
        ${\bf{D}}^{q}$, ${\bf{X}}$ & anchor knots and window bound values~\eqref{eq:anchors}\\
        $\per$, $\sigma_{\mathrm{tol}}$ & window factor (Section~\ref{subsec:certified-knots}), rank-test tolerance\\
        ${\bf{K}}$ & stiffness matrix~\eqref{eq:energy}\\
        \bottomrule
    \end{tabular}
\end{table}

\subsection{Periodic B-splines and the collocation matrix}
\label{subsec:periodic}

Let ${\bf{D}} = \{0 = u_0 < u_1 < \cdots < u_{\hat n+1} = 1\}$ be \emph{domain knots}, extended periodically by $u_{j + (\hat n + 1)} = u_j + 1$ for all $j \in \mathbb{Z}$; throughout, $\hat n \ge q \ge 1$.
For $i = 0, \ldots, \hat n$ let $N^{\mathrm{line}}_{i}$ denote the ordinary B-spline of degree $q$ on the $q+2$ consecutive knots
\[
    (u_{i-k}, \ldots, u_{i+k}) \;\; \text{for odd } q = 2k-1,
    \qquad
    (u_{i-k}, \ldots, u_{i+k+1}) \;\; \text{for even } q = 2k,
\]
taken from the periodic extension, so that $N^{\mathrm{line}}_{i}$ is centered at $u_i$ (odd $q$) or at the span $[u_i, u_{i+1}]$ (even $q$).
The \emph{periodic basis} consists of the periodizations
\begin{equation}
    \label{eq:per-basis}
    N^{\mathrm{per}}_{i,q,{\bf{D}}}(u) \;=\; \sum_{\ell \in \mathbb{Z}} N^{\mathrm{line}}_{i}(u + \ell),
    \qquad u \in [0, 1),\;\; i = 0, \ldots, \hat n ,
\end{equation}
which satisfy the partition of unity $\sum_i N^{\mathrm{per}}_{i,q,{\bf{D}}} \equiv 1$ and are $C^{q-1}$ across the seam.
Since $\hat n \ge q$, at most two terms of~\eqref{eq:per-basis} are nonzero at any $u$.
The basis coincides with the folded basis of~\cite{tang2023control} up to a cyclic index shift of $\lfloor q/2 \rfloor$; the shift is chosen so that $N^{\mathrm{per}}_{i}$ is associated with the knot $u_i$.

A section is the periodic B-spline curve on the domain knots ${\bf{D}}$ with control-point matrix ${\bf{P}} \in \mathbb{R}^{(\hat n + 1) \times 3}$:
\[
    {\bf{C}}_{\hat n}(u) \;=\; \sum_{i=0}^{\hat n} N^{\mathrm{per}}_{i,q,{\bf{D}}}(u)\, {\bf{P}}_i ,
    \qquad 0 \le u < 1 .
\]

Given parameters ${\bf{T}} = \{0 \le t_0 < \cdots < t_n < 1\}$, the \emph{periodic} collocation matrix is defined by
\begin{equation}
    \label{eq:collocation}
    {\bf{N}}({\bf{D}}, {\bf{T}})_{j,i} \;=\; N^{\mathrm{per}}_{i,q,{\bf{D}}}(t_j)
    \;\in\; \mathbb{R}^{(n+1) \times (\hat n + 1)} ,
\end{equation}
square when $\hat n = n$ and rectangular (under-determined) when $\hat n > n$.
Interpolating the data points ${\bf{Q}}_j$ at the parameters $t_j$, i.e. requiring ${\bf{C}}_{\hat n}(t_j) = {\bf{Q}}_j$, is the linear system
\begin{equation}
    \label{eq:interp-system}
    {\bf{N}}({\bf{D}}, {\bf{T}})\, {\bf{P}} = {\bf{Q}} .
\end{equation}

The solvability claims of the earlier paper rest on the following condition (Theorems~2--3 of~\cite{tang2023control}, 
stated here in the centered indexing; knot indices outside $\{0, \ldots, \hat n + 1\}$ refer to the periodic extension).

\begin{quote}
    \hrule
    \emph{Periodic window condition.} With $k = \lceil q/2 \rceil$:
    \begin{equation}
        \label{eq:window}
        u_{j-k} < t_j < u_{j+k} \;\; (q = 2k - 1 \text{ odd}),
        \quad
        u_{j-k} < t_j < u_{j+k+1} \;\; (q = 2k \text{ even}),
        \quad j = 0, \ldots, n . \quad
    \end{equation}
    With the centered indexing above and the knots and parameters taken on the same cyclic lift, 
    this is equivalent to positivity of every diagonal entry of the square collocation matrix: $N^{\mathrm{per}}_{j,q,{\bf{D}}}(t_j) > 0$.
    \vspace{1mm}
    \hrule
\end{quote}

This is the natural periodic analogue of the classical Schoenberg--Whitney condition, which characterizes nonsingularity for B-splines on a line~\cite{de1978practical}.
The earlier paper stated that the window condition implies nonsingularity of the square periodic matrix, 
and conjectured that the rectangular matrix~\eqref{eq:collocation} has full row rank whenever a knot subsequence satisfies the condition.
Section~\ref{sec:theory} separates the valid odd-cardinality case from the even-cardinality failure and develops a square-subsystem condition for the rectangular system.

\section{Solvability of periodic collocation}
\label{sec:theory}

\subsection{Singular configurations that satisfy the window condition}
\label{subsec:insufficiency}

Consider first degree $q = 1$ on uniform knots $u_j = j/(n+1)$, collocated at the interval midpoints $t_j = (j + \tfrac12)/(n+1)$.
Every row of the collocation matrix is $(\ldots, 0, \tfrac12, \tfrac12, 0, \ldots)$ and every diagonal entry equals $\tfrac12 > 0$, so the window condition holds.
If $n+1$ is even, the alternating sign pattern $(1, -1, \ldots, 1, -1)$ is consistent around the circle, and every row of the matrix annihilates it: the matrix is singular.
If $n+1$ is odd, the alternation cannot close up, and the matrix is nonsingular.
The same mechanism operates at other degrees.

\begin{proposition}[Singular configurations]
    \label{prop:singular}
    We give two counterexamples, one on uniform knots and one on non-uniform knots.

    \noindent \emph{(i) Uniform.} For $n+1=12$, the uniform knots $u_j = j/(n+1)$ and the collocation parameters
    \[
        t_j = \frac{j+1/2}{n+1} \;\; (q = 3),
        \qquad
        t_j = \frac{j}{n+1} \;\; (q = 2),
        \qquad j = 0, \ldots, n,
    \]
    satisfy the periodic window condition \eqref{eq:window}, yet the collocation matrix ${\bf{N}}({\bf{D}},{\bf{T}})$ is singular.

    \noindent \emph{(ii) Non-uniform.} Let $q = 3$, $n + 1 = 12$, and let ${\bf{D}}$ be the rational domain knots with gap proportions
    \[
        (87,\, 45,\, 71,\, 52,\, 63,\, 90,\, 41,\, 77,\, 58,\, 66,\, 83,\, 49)/782 ,
    \]
    a non-uniform knot sequence whose largest-to-smallest gap ratio is about $2.2$.
    For $s \in [0,1]$ define $t_j(s) = u_j + s\,(u_{j+1} - u_j)$.
    Then $({\bf{D}}, {\bf{T}}(s))$ satisfies the window condition for every $s \in [\tfrac12, \tfrac{257}{512}]$, and
    \[
        \det {\bf{N}}\bigl({\bf{D}}, {\bf{T}}(\tfrac12)\bigr) \;<\; 0 \;<\; \det {\bf{N}}\bigl({\bf{D}}, {\bf{T}}(\tfrac{257}{512})\bigr),
    \]
    both signs evaluated in exact rational arithmetic.
    Consequently there exists $s^* \in (\tfrac12, \tfrac{257}{512})$ at which ${\bf{N}}({\bf{D}}, {\bf{T}}(s^*))$ is singular while the window condition holds.
\end{proposition}

\begin{proof}
    For (i), the validation is straightforward.
    For (ii) the two determinant signs were computed in exact rational arithmetic by running the Cox--de~Boor recursion in a computer algebra system (Wolfram Mathematica\textsuperscript{\textregistered}); the same system was used to check the endpoint window inequalities exactly.
    For fixed ${\bf{D}}$ the window condition consists of strict inequalities that are linear in ${\bf{T}}$, and ${\bf{T}}(s)$ is affine in $s$; both endpoints satisfy the inequalities, hence so does every intermediate $s$.
    The determinant is continuous in $s$, so it vanishes somewhere in between.
\end{proof}

Proposition~\ref{prop:singular} shows that Theorems~2 and~3 of~\cite{tang2023control} do not hold as stated, and that Proposition~1 there (conjectured by Tang) needs a stronger hypothesis.
Theorem~\ref{thm:full-rank-lift} below gives the corresponding corrected statement.
Proposition~\ref{prop:singular} refutes the assertion of Meir and Sharma~\cite{meir1968convergence} at every even point count; it also shows that the failure persists on non-uniform knots.

\begin{remark}[Scope of the failure]
    \label{rem:diagnosis}
    The proof in~\cite{tang2023control} transfers the Schoenberg--Whitney equivalence from B-splines on a line to the periodic matrix.
    The transfer fails because periodization~\eqref{eq:per-basis} sums translates of B-splines, and the resulting functions no longer have the ordered support structure on which the classical zero-counting argument rests.
    The underlying obstruction is a parity constraint: a Haar space of continuous functions on the circle must have odd dimension, so every even-dimensional periodic spline space admits some singular node configuration.
    Proposition~\ref{prop:singular} locates such configurations inside the proposed window region.
    This disproves the window condition at even cardinality.
\end{remark}

For odd cardinality the situation is different and the periodic window condition is sufficient. We record this formally in the following theorem.

\begin{theorem}[Odd-cardinality window condition]
    \label{thm:odd-window}
    Let $\hat n=n\ge q$ and suppose that $n+1$ is odd.
    If $({\bf D},{\bf T})$ satisfies the periodic window condition \cite{tang2023control}, then ${\bf N}({\bf D},{\bf T})$ is nonsingular.
\end{theorem}

\begin{proof}
    Cut the parameter circle at a point congruent to neither a knot nor a parameter, and lift one period to an open interval containing the parameters.
    The resulting $(n+1)$-dimensional periodic spline space is weak Chebyshev when $n+1$ is odd~\cite{schumaker2007spline,zeilfelder1996periodic}.
    Its periodic B-spline basis is locally linearly independent, and it has no common zero because the basis is a partition of unity.
    The support form of the Schoenberg--Whitney theorem for such spaces states that a set of $n+1$ nodes is an interpolation set if the nodes can be bijectively assigned to the interiors of the basis supports~\cite[Theorem~3.11]{davydov2000weak}.
    The window condition gives this assignment directly: $t_j\in\operatorname{int}\operatorname{supp}N^{\mathrm{per}}_{j,q,{\bf D}}$ for every $j$.
    Hence the evaluation map is bijective and ${\bf N}({\bf D},{\bf T})$ is nonsingular.
\end{proof}

\subsection{Anchor knots and perturbation intervals}
\label{subsec:sufficient}

For closed B-spline interpolation the knots cannot simply be averaged from the parameters, because the collocation matrix carries additional wrap-around constraints. 
Park and Kim~\cite{park1996smooth} gave a knot-placement rule that follows the parameter distribution (called the \emph{average} method) but is reliable only for odd degrees; 
Park's \emph{natural} and \emph{shifting} methods~\cite{park2001choosing} instead set the knots to the parameters or to the midpoints of adjacent parameter spans. 
The \emph{natural} method applies to odd degree and the \emph{shifting} method to even degree. 
Each of these placements determines the knot vector uniquely and produces a smooth interpolation curve that follows the polyline of the corresponding data points.
We summarize all the related methods for knot determination as follows:
\begin{equation}
    u_j =
    \begin{cases}
        d_n / 2 + (t_{j-1} + t_{j}) / 2 = (\tilde t_{j-1} + \tilde t_{j}) / 2 & {\text{shifting \ }} (j = 1, \cdots, n)  \\
        t_{j}                            & {\text{natural}}                      \\
        (t_{j-1} + t_{j} + t_{j+1}) / 3  & {\text{average}}                    \\
    \end{cases}
\end{equation}
For even $q$ they are collocated at the \emph{shifted} parameters $\tilde t_j = t_j + d_n/2$, and for odd $q$ they are collocated at the parameters themselves. 
All of these knot-placement strategies are independent of the parity of $n+1$.

We treat such a placement as an \emph{anchor} configuration, and the anchor knots may themselves be slightly perturbed, 
giving an interpolation curve close to that of the ideal anchors. The same idea underlies 
the open-curve and surface reduction methods of Piegl and Tiller~\cite{piegl2000reducing}, Park~\cite{park2003lofted}, and Wang~\etal\cite{wang2008reducing}. 
In order to obtain the interpolation curve, a perturbed anchor must still be solvable.
For open-curve interpolation, solvability is guaranteed by the Schoenberg--Whitney condition~\cite{de1978practical}. 
Figures~\ref{fig:cubic-anchor-perturbation} and~\ref{fig:quadratic-anchor-perturbation} illustrate
representative anchor-knot perturbations for natural interpolation of odd degree and shifting interpolation of even degree,
using two distinct non-uniform eight-point samples of the unit circle.
In each figure, the interpolants, independent control points, and domain knots are separated from
left to right, with the knot displacements magnified tenfold.
In both examples, the perturbed systems remain nonsingular and the resulting interpolants remain
close to the anchor curves.
Both figures use $n+1=8$, an even point count at which the window condition may fail; yet the
perturbations, taken within the window, do not produce a singular system, illustrating that the
singular configurations of Proposition~\ref{prop:singular} are exceptional.

\begin{figure}[!t]
    \centering
    \includegraphics[width=0.98\textwidth]{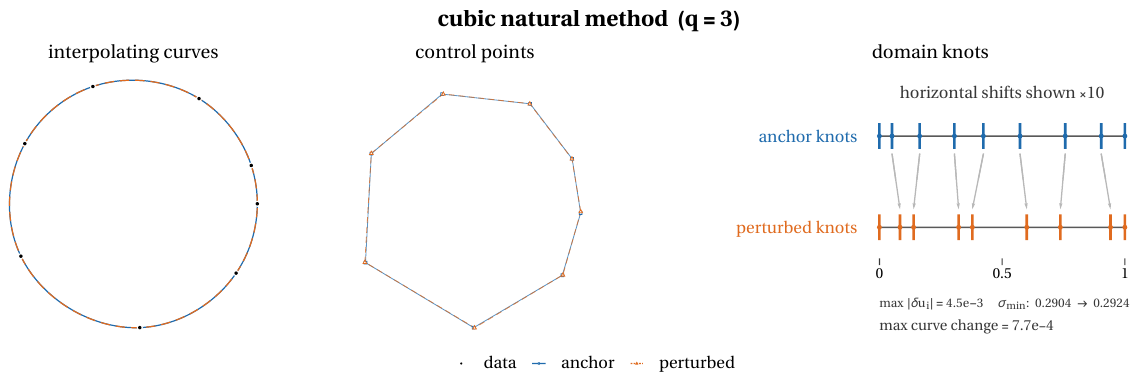}
    \caption{Cubic natural interpolation under anchor-knot perturbation ($n+1=8$ even, $q=3$).}
    \label{fig:cubic-anchor-perturbation}
\end{figure}

\begin{figure}[!t]
    \centering
    \includegraphics[width=0.98\textwidth]{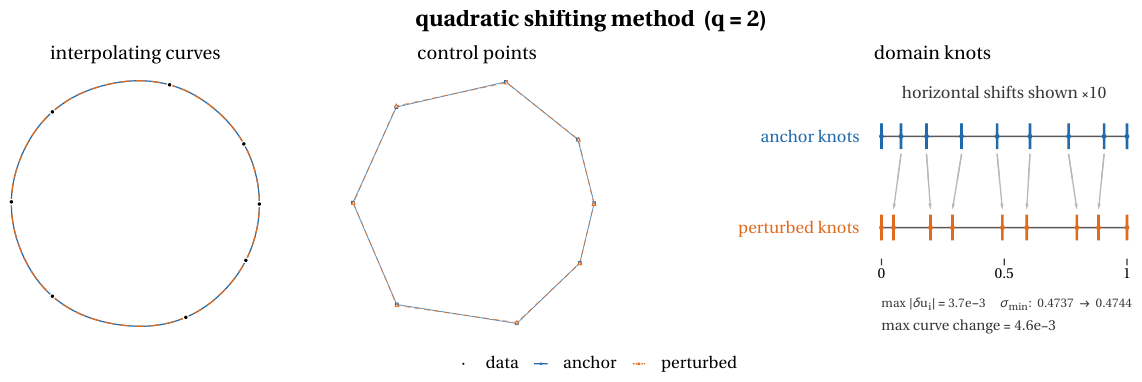}
    \caption{Quadratic shifting interpolation under anchor-knot perturbation ($n+1=8$ even, $q=2$).}
    \label{fig:quadratic-anchor-perturbation}
\end{figure}

Tang \cite{tang2023control} extended the open-curve reduction scheme of Piegl and Tiller~\cite{piegl2000reducing} to closed contours and, 
in place of the open-curve Schoenberg--Whitney test, defined a \emph{perturbation interval} for each knot of a closed section.
Figure~\ref{fig:anchor-window} illustrates the construction: the anchor knot $u_j^{q}$ and the two adjacent bound values $x_{j-1}$ and $x_j$ delimit a half-open window from which the $j$-th knot is selected, 
with the window factor $\per$ scaling the interval between the anchor and the bound values.
In the odd case, the anchors coincide with the parameters and the bounds are the parameter midpoints; 
in the even case, the anchors are the midpoints of the shifted parameters and the bounds coincide with the shifted parameters themselves.

\begin{figure}[htbp]
    \centering
    \includegraphics[width=0.7\textwidth]{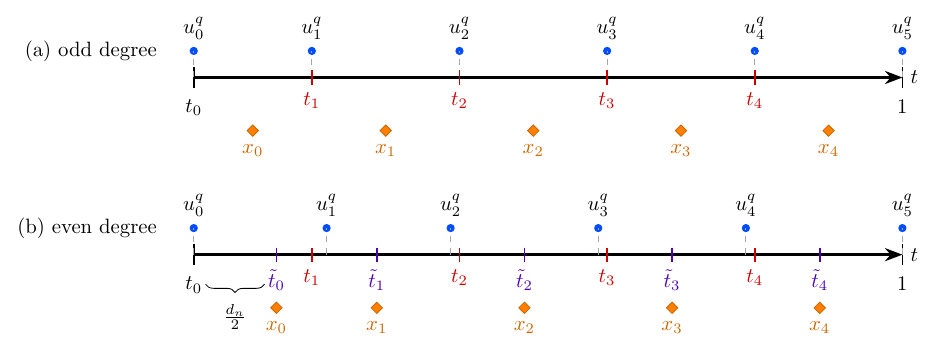}
    \caption{Anchor knots and bound values defining the perturbation interval $[x_{j-1}, x_j)$ of a closed section: (a) odd degree, with anchors $u_j^{q}=t_j$ and bounds $x_j=(t_j+t_{j+1})/2$; (b) even degree, with anchors $u_j^{q}=(\tilde t_{j-1}+\tilde t_j)/2$ and bounds $x_j=\tilde t_j$.}
    \label{fig:anchor-window}
\end{figure}

\subsection{Full row rank lifting}
\label{subsec:full-rank-lift}

The rectangular case, in which one section is interpolated on the common knots of all sections, reduces to the square case by nestedness of periodic spline spaces.
Write ${\bf{b}}_{{\bf{D}}}(u)$ for the column vector
$\bigl(N^{\mathrm{per}}_{i,q,{\bf{D}}}(u)\bigr)_{i=0}^{\hat n}$
and $\mathbb{S}_{q,{\bf{D}}} = \operatorname{span}\{N^{\mathrm{per}}_{i,q,{\bf{D}}} \mid i = 0, \dots, \hat{n}\}$;
for a subsequence $\bar{{\bf{D}}}$ of ${\bf{D}}$, the barred symbols $\bar{{\bf{b}}}(u)$ and $\mathbb{S}_{q,\bar{{\bf{D}}}}$
denote the corresponding basis vector and spline space built on $\bar{{\bf{D}}}$.
The following lemma records the nesting property on which this reduction rests.

\begin{lemma}[Nested periodic spline spaces]
    \label{lem:nested}
    Let $\hat n \ge q \ge 1$.
    The functions $\{N^{\mathrm{per}}_{i,q,{\bf{D}}}\}_{i=0}^{\hat n}$ are locally linearly independent on the circle and 
    span the space of $1$-periodic $C^{q-1}$ piecewise polynomials of degree $\le q$ with breakpoints in the periodic extension of ${\bf{D}}$, of dimension $\hat n + 1$~\cite[\S 8.1]{schumaker2007spline}.
    Moreover, if $\bar{{\bf{D}}} \subset {\bf{D}}$ is a subsequence containing $0$ and $1$, of size $n + 2$ with $q \le n < \hat n$, 
    then $\mathbb{S}_{q,\bar{{\bf{D}}}} \subset \mathbb{S}_{q,{\bf{D}}}$ and there is ${\bf{M}} \in \mathbb{R}^{(n+1) \times (\hat{n}+1)}$ with $\bar{{\bf{b}}}(u) = {\bf{M}}\,{\bf{b}}_{{\bf{D}}}(u)$ on $[0,1)$ and $\operatorname{rank} {\bf{M}} = n+1$.
\end{lemma}

The proof is given in Appendix~\ref{app:proofs}.

The nestedness recorded in Lemma~\ref{lem:nested} is what the energy-minimizing lofting of Section~\ref{sec:certified} ultimately relies on.
There each section is interpolated on the common knots shared by all sections, so the collocation system~\eqref{eq:interp-system} is under-determined ($\hat n > n$); 
full row rank of ${\bf{N}}$ is then exactly the property that guarantees feasibility for closed section data and a unique energy-minimizing interpolant.
The following theorem shows that this property can be certified from a single nonsingular square subsystem.

\begin{theorem}[Full row rank lifting]
    \label{thm:full-rank-lift}
    Let $\hat{n} > n \ge q$, let ${\bf{D}}$ be domain knots of size $\hat n+2$ and ${\bf{T}}$ parameters of size $n+1$.
    If some subsequence $\bar{{\bf{D}}} = \{0 = \bar u_0 < \cdots < \bar u_{n+1} = 1\} \subset {\bf{D}}$ has nonsingular square collocation matrix $\bar{{\bf{N}}} = {\bf{N}}(\bar{{\bf{D}}}, {\bf{T}})$, then the rectangular ${\bf{N}}({\bf{D}}, {\bf{T}})$ has full row rank.
\end{theorem}

\begin{proof}
    Row $j$ of $\bar{{\bf{N}}}$ is $\bar{{\bf{b}}}(t_j)^{\mathsf{T}} = {\bf{b}}_{{\bf{D}}}(t_j)^{\mathsf{T}}{\bf{M}}^{\mathsf{T}}$ by Lemma~\ref{lem:nested}, and ${\bf{b}}_{{\bf{D}}}(t_j)^{\mathsf{T}}$ is row $j$ of ${\bf{N}}$; hence $\bar{{\bf{N}}} = {\bf{N}}{\bf{M}}^{\mathsf{T}}$ and
    \[
        n+1 = \operatorname{rank} \bar{{\bf{N}}} = \operatorname{rank}({\bf{N}}{\bf{M}}^{\mathsf{T}}) \le \operatorname{rank} {\bf{N}} \le n+1,
    \]
    so \(\operatorname{rank} {\bf{N}} = n+1\).
\end{proof}

Theorem~\ref{thm:full-rank-lift} is the corrected form of Proposition~1 of~\cite{tang2023control}.
At odd cardinality the window hypothesis is sufficient by Theorem~\ref{thm:odd-window}; 
at even cardinality Proposition~\ref{prop:singular} shows that it must be replaced by nonsingularity of $\bar{{\bf N}}$ itself.
This hypothesis is guaranteed by the anchor knots defined by natural and shifting methods~\cite{park2001choosing} or can be checked operationally by $\sigma_{\min}(\bar{{\bf N}})>\sigma_{\mathrm{tol}}$ on the square collocation matrix $\bar{{\bf{N}}}$.
The mathematical statements use the exact singular value; Section~\ref{sec:experiments} describes the floating-point safeguard used in the implementation.

The preceding results supply the safe square subsystems and rank transfer needed by the constrained energy formulation.

\section{Energy-minimizing lofting framework}
\label{sec:certified}

\subsection{Energy minimization and saddle-point solvability}
\label{subsec:energy}

For $\hat n > n$, the system~\eqref{eq:interp-system} is under-determined, and the interpolant of minimal strain energy is selected:
\begin{equation}
    \label{eq:energy}
    E({\bf{C}}_{\hat{n}}) = \int_{0}^{1} \bigl( \alpha \| \dot{{\bf{C}}}_{\hat n} \|^2 + \beta \| \ddot{{\bf{C}}}_{\hat n} \|^2 \bigr) \mathrm{d}u
    = \operatorname{tr}\bigl( {\bf{P}}^{\mathsf{T}} {\bf{K}} {\bf{P}} \bigr),
\end{equation}
where ${\bf{K}} \in \mathbb{R}^{(\hat{n}+1) \times (\hat{n}+1)}$ is the stiffness matrix defined by
\[
    {\bf{K}} = \int_{0}^{1} \bigl( \alpha \, {\bf{b}}^{(1)} {\bf{b}}^{(1)\mathsf{T}} + \beta \, {\bf{b}}^{(2)} {\bf{b}}^{(2)\mathsf{T}} \bigr) \mathrm{d}u,
\] 
and ${\bf{b}} = {\bf{b}}_{{\bf{D}}}$ is the periodic basis vector; 
$\alpha \ge 0$ and $\beta \ge 0$ are the stretching and bending coefficients, respectively~\cite{park2001choosing,park2003lofted}.
Here and below, ${\bf{N}}$ abbreviates the periodic collocation matrix ${\bf{N}}({\bf{D}}, {\bf{T}})$.
The Lagrange system for~\eqref{eq:energy} under the interpolation constraints~\eqref{eq:interp-system} is
\begin{equation}
    \label{eq:saddle}
    \begin{pmatrix}
        {\bf{K}} & {\bf{N}}^{\mathsf{T}} \\[2pt]
        {\bf{N}}  & {\bf{0}}
    \end{pmatrix}
    \begin{pmatrix}
        {\bf{P}} \\[2pt] {\bf{v}}
    \end{pmatrix}
    =
    \begin{pmatrix}
        {\bf{0}} \\[2pt] {\bf{Q}}
    \end{pmatrix}.
\end{equation}

For $\alpha > 0$, $\beta \ge 0$ (or $\alpha \ge 0$, $\beta > 0$, $q \ge 2$), 
the matrix in~\eqref{eq:saddle} is nonsingular if ${\bf{N}}$ has full row rank and yields a unique energy-minimizing interpolating curve~\cite{wang2008reducing,benzi2005numerical}.
Theorem~\ref{thm:full-rank-lift} shows how this full row rank is guaranteed: it suffices that each section retain a nonsingular square subsystem, 
whose rank transfers to the rectangular common-knot matrix.

\subsection{Knot construction and rank checks}
\label{subsec:certified-knots}

This subsection turns the rank conditions of Section~\ref{sec:theory} into a per-section knot-construction procedure.
The periodic window condition~\cite{tang2023control} certifies solvability at odd point counts (Theorem~\ref{thm:odd-window}); at even point counts it is not sufficient, 
but the counterexamples of Proposition~\ref{prop:singular} are not easily constructed and are rarely encountered in practice.
A rank check is therefore applied as a safeguard: at even point counts the selected knots are verified by a singular-value test, 
and the nonsingular anchor knots provide a safe fallback whenever the test fails.

Procedure~\ref{lst:proc1} selects, for one section with collocation parameters ${\boldsymbol\tau}=\{\tau_j\}_{j=0}^n$, a square subsystem from input domain knots ${\bf D}$ with $\hat n\ge n$.
Here $\tau_j=t_j$ for odd $q$ and $\tau_j=\tilde t_j$ for even $q$, with $\tau_{n+1}=\tau_0+1$.
The selection reuses the anchor knots and bound values of~\cite{tang2023control}.
In terms of ${\boldsymbol\tau}$, the \emph{anchor knots} ${\bf D}^{q}=\{u_0^q,\ldots,u_{n+1}^q\}$ and \emph{bound values} ${\bf X}=\{x_0,\ldots,x_n\}$ are
\begin{equation}
    \label{eq:anchors}
    u_j^{q} =
    \begin{cases}
        \tau_j, & q \text{ odd},\\[4pt]
        \dfrac{\tau_{j-1}+\tau_j}{2}, & q \text{ even},
    \end{cases}
    \qquad
    x_j =
    \begin{cases}
        \dfrac{\tau_j+\tau_{j+1}}{2}, & q \text{ odd},\\[4pt]
        \tau_j, & q \text{ even},
    \end{cases}
\end{equation}
with $u_0^{q} = 0$ and $u_{n+1}^{q} = 1$ (see Figure~\ref{fig:anchor-window}); the anchor knots are exactly the natural and shifting configurations (Section~\ref{subsec:sufficient}).
The $j$-th knot may be taken anywhere in the half-open window:
\[
    [a_j, b_j) = \bigl[\, (1-\per)\,u_j^{q} + \per\, x_{j-1},\; (1-\per)\,u_j^{q} + \per\, x_j \,\bigr) ;
\]
the first input knot in the window is taken if one exists, and the anchor $u_j^{q}$ otherwise.
The factor $\per \in [0, 1]$ scales the window widths.
At $\per=1.0$ the windows are widest and favor knot sharing; at $\per=0$ every window collapses to its anchor knot, and each selection returns the section anchors.
Figure~\ref{fig:knot-flex} illustrates this window geometry for a single knot.

\begin{figure}[htbp]
    \centering
    \includegraphics[width=0.5\textwidth]{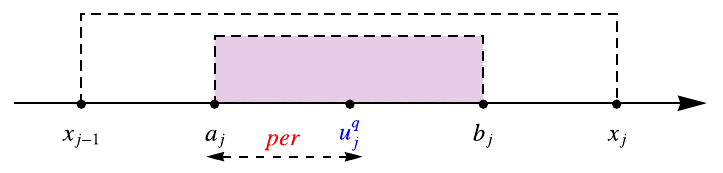}
    \caption{Window geometry of the per-section knot selection. The $j$-th knot is chosen in the half-open window $[a_j, b_j)$ between the anchor knot $u_j^{q}$ and the bound values $x_{j-1}$ and $x_j$, scaled by the window factor $\per$.}
    \label{fig:knot-flex}
\end{figure}

We call the sequence produced by this window-selection loop, before the final parity branch, the \emph{first-fit candidate}.
Because the candidate is chosen inside the selection windows, it satisfies the periodic window condition; 
at odd cardinality Theorem~\ref{thm:odd-window} therefore certifies that its square collocation matrix is nonsingular, and the candidate is returned without any further test.
At even cardinality the window condition is no longer sufficient, so the candidate is subjected to the rank test
\[
    \sigma_{\min}\bigl( {\bf{N}}(\bar{{\bf{D}}}, {\boldsymbol\tau}) \bigr) > \sigma_{\mathrm{tol}}
    \qquad (\text{we use } \sigma_{\mathrm{tol}} = 10^{-8}),
\]
a singular-value computation on a square matrix; if the test fails, the section falls back to its anchor knots.
For odd and even degrees, respectively, the fallback collocation matrices are nonsingular on given collocation parameters $\boldsymbol\tau$, 
as guaranteed by the natural and shifting methods~\cite{park2001choosing}.
Procedure~\ref{lst:proc1} therefore terminates for every input without a retry loop and returns a nonsingular square subsystem.
The pseudocode is given in Listing~\ref{lst:proc1}.

\vspace{2mm}
\begin{lstlisting}[mathescape, escapechar=\%, caption={Procedure 1: per-section knot selection with a rank check.}, label={lst:proc1}]
  Procedure 1 $\mathtt{select\_knots}({\boldsymbol\tau}, {\bf{D}}, q, \per, \sigma_{\mathrm{tol}})$
  INPUT:
  $\quad {\boldsymbol\tau} = \{\tau_0, \ldots, \tau_n\} \gets$ collocation parameters (shifted if $q$ even)
  $\quad {\bf{D}} = \{u_0, \ldots, u_{\hat{n}+1}\} \gets$ input domain knots, $\hat n\ge n$
  $\quad q,\ \per \gets$ section degree (odd or even), window factor
  $\quad \sigma_{\mathrm{tol}} \gets$ rank-test tolerance
  OUTPUT:
  $\quad \bar{{\bf{D}}} = \{\bar{u}_0, \ldots, \bar{u}_{n+1}\} \gets$ selected domain knots making square subsystem nonsingular
  ALGORITHM:
    ${\bf{D}}^{q}, {\bf{X}} \gets$ anchor knots and bound values, see %Eq.~\eqref{eq:anchors}%
    $\bar{u}_0 \gets 0$;  $\bar{u}_{n+1} \gets 1$;  $\mathrm{span} \gets 0$
    for $j = 1$ to $n$ do  // first-fit selection in the windows
      $a \gets (1 - \per) \cdot u_j^{q} + \per \cdot x_{j-1}$
      $b \gets (1 - \per) \cdot u_j^{q} + \per \cdot x_j$
      while $u_{j + \mathrm{span}} < a$ do  
        $\mathrm{span} \mathrel{+}= 1$
        if $j + \mathrm{span} > \hat{n} + 1$ then  // input knots exhausted:
          $\bar{u}_l \gets u_l^{q}$ for all remaining $l$;  exit the for-loop
        end if
      end while
      if $u_{j + \mathrm{span}} < b$ then  $\bar{u}_j \gets u_{j + \mathrm{span}}$
      else $\bar{u}_j \gets u_j^{q}$;  $\mathrm{span} \mathrel{-}= 1$
    end for
    if $n+1$ is odd then return $\bar{{\bf{D}}}$  // window condition holds; nonsingular by Theorem %\ref{thm:odd-window}%
    // n + 1 is even; check the rank of the candidate
    else if $\sigma_{\min}({\bf{N}}(\bar{{\bf{D}}}, {\boldsymbol\tau})) > \sigma_{\mathrm{tol}}$ then return $\bar{{\bf{D}}}$  // rank test passed
    else return ${\bf{D}}^{q}$  // nonsingular by anchor knots, see Section %\ref{subsec:sufficient}% 
\end{lstlisting}

Procedure~\ref{lst:proc2} builds the common domain knots as in~\cite{wang2008reducing}.
The initial knots are computed from the averaged parameters of the largest sections.
Each section then selects a square subsystem from the current knots with Procedure~\ref{lst:proc1}, and this subsequence is merged into the common set by sorted union with duplicates removed.

\vspace{2mm}
\begin{lstlisting}[mathescape, escapechar=\%, caption={Procedure 2: common domain knots with per-section square subsystems.}, label={lst:proc2}]
  Procedure 2 $\mathtt{build\_common\_knots}(\mathcal{Q}, q, \per, \sigma_{\mathrm{tol}})$
  INPUT:
  $\quad \mathcal{Q} = \{{\bf{Q}}_{i,j}\} \gets$ sections of data points ($i = 0, \ldots, m$; $j = 0, \ldots, n_i$)
  $\quad q,\ \per \gets$ section degree (odd or even), window factor
  $\quad \sigma_{\mathrm{tol}} \gets$ rank-test tolerance
  OUTPUT:
  $\quad {\bf{D}} \gets$ common domain knots
  $\quad \bar{{\bf{D}}}^{(0)}, \ldots, \bar{{\bf{D}}}^{(m)} \gets$ selected domain knots of every section
  ALGORITHM:
    $n_{\max} \gets \max_i n_i$
    ${\boldsymbol\tau}_{\mathrm{avg}} \gets$ mean collocation-parameter vector of sections with $n_i = n_{\max}$
    ${\bf{D}} \gets$ anchor knots of ${\boldsymbol\tau}_{\mathrm{avg}}$, see %Eq.~\eqref{eq:anchors}%  // initial knots ${\bf{D}}_0$
    for $i = 0$ to $m$ do
      ${\boldsymbol\tau}^{(i)} \gets$ collocation parameters of section $i$ (shifted if $q$ even)
      $\bar{{\bf{D}}}^{(i)} \gets \mathtt{select\_knots}({\boldsymbol\tau}^{(i)}, {\bf{D}}, q, \per, \sigma_{\mathrm{tol}})$
      ${\bf{D}} \gets \mathtt{merge}({\bf{D}}, \bar{{\bf{D}}}^{(i)})$  // sorted union
    end for
    return ${\bf{D}}$, $\{\bar{{\bf{D}}}^{(i)}\}_{i=0}^{m}$
\end{lstlisting}

Because Procedure~\ref{lst:proc2} retains the initial knots obtained from the averaged largest-section parameters, 
its output at $\per=0$ need not equal the traditional union when several largest sections are averaged.

\begin{proposition}[Solvability of the constructed knots]
    \label{prop:proc2}
    Let $q\ge2$, and suppose that $\alpha>0$ and $\beta\ge0$ (or $\alpha\ge0$ and $\beta>0$).
    Then Procedure~\ref{lst:proc1} returns, for every section, a square subsystem whose collocation matrix is nonsingular; Procedure~\ref{lst:proc2} terminates; 
    and on its output ${\bf{D}}$, every section's collocation matrix has full row rank and the saddle-point system~\eqref{eq:saddle} is uniquely solvable.
\end{proposition}

\begin{proof}
    Procedure~\ref{lst:proc1} returns either the first-fit candidate or the anchor knots.
    At odd cardinality the candidate is nonsingular by Theorem~\ref{thm:odd-window}; at even cardinality it is returned only if the rank test passes, 
    and otherwise the anchors are returned, which are nonsingular by the natural method for odd degree and the shifting method for even degree (Section~\ref{subsec:sufficient}).
    Procedure~\ref{lst:proc2} performs one such call per section and a finite merge, so it terminates.
    Merging only adds knots, and each section's square subsystem is merged into the evolving set; hence every subsystem is a subsequence of the final ${\bf{D}}$, 
    and its collocation matrix depends only on the subsystem and the section's parameters.
    If the final ${\bf D}$ equals a section's subsystem, its matrix is already nonsingular; if it is a strict refinement, Theorem~\ref{thm:full-rank-lift} gives full row rank.
\end{proof}

With the common knots in hand, the surface is completed as usual: system~\eqref{eq:saddle} is solved for every section on the common domain knots ${\bf{D}}$, 
and the control-point columns are interpolated longitudinally (Algorithm~A9.4 of~\cite{piegl1996nurbs}).
Because every section is interpolated on the same common knots ${\bf{D}}$, the stiffness matrix ${\bf{K}}$ of~\eqref{eq:energy} is assembled once and reused for all sections.
Unlike the closed-contour lofting algorithm of Tang~\cite{tang2023control}, which relies on clamping and knot refinement, the present method requires neither clamping nor knot refinement.
This completes the energy-minimizing lofting construction for closed contours.

\section{Experimental results}
\label{sec:experiments}

The proposed method (Procedures~\ref{lst:proc1} and~\ref{lst:proc2}) is implemented in C and the Wolfram Language.
The experimental setup follows~\cite{tang2023control}: the longitudinal degree is $p=3$; the section degree takes the values $q=2,3,4,5$; 
the strain-energy weights are $\alpha=1$ and $\beta=0.2$, as recommended in~\cite{park2003lofted}; and the window factor is $\per=1.0$ unless stated otherwise.
Singular values are computed in IEEE double precision on the unscaled, row-stochastic collocation matrices and compared with the absolute tolerance $\sigma_{\mathrm{tol}}=10^{-8}$.
All comparisons use the same aligned sections, seam locations, chord-length parameters, and curve degrees.
The traditional closed-lofting baseline interpolates each section on its own periodic knots and refines the section curves to a common knot vector formed from their union.
The knot-selection baseline is the closed-contour method of~\cite{tang2023control} with the same degree and window factor.
To correct its knot selection, the same mechanism as in Procedure~\ref{lst:proc1} is applied:
for sections with an even point count, the selected knots are verified by the rank check, and a failed check falls back to the anchor knots.
The \emph{Proposed} method is the present rank-checked, section-wise energy-minimizing construction.
A reported net size $a\times b$ consists of $a$ longitudinal rows and $b$ closed-section control points, including the $q$ wrap-around control points.

The evaluation uses six case-study models: a water bottle (12 sections, 54--80 points), a shoe last (10 sections, 69--101 points),
an aero-engine blade (11 sections, 140--160 points), a human jaw (9 sections, 111--136 points),
a liver (48 sections, 40--90 points), and a human head (92 sections, 90--179 points), for 182 section curves in total.
Because the rank check is parity-aware, a section with an odd point count is certified structurally by the window condition (Theorem~\ref{thm:odd-window}), 
whereas an even point count requires the rank test.
Table~\ref{tab:parity} reports, for each model, how many sections have an odd and an even point count.

Figure~\ref{fig:surfaces} shows the four small models, whose row counts are modest.
For each model, the left panel shows the data points and the polyline through the section start points.
The middle panel is an open visualization of the same reconstruction, included to expose the seam; the right panel shows the closed surface produced by the present method.
Figures~\ref{fig:liver} and~\ref{fig:head} show the liver and head reconstructions with their control nets. 
The proposed approach achieves more efficient data reduction (Table~\ref{tab:counts-full}) than the traditional approach while providing visually pleasing B-spline surfaces.

\begin{figure}[htbp]
    \centering
    \includegraphics[width=0.9\textwidth]{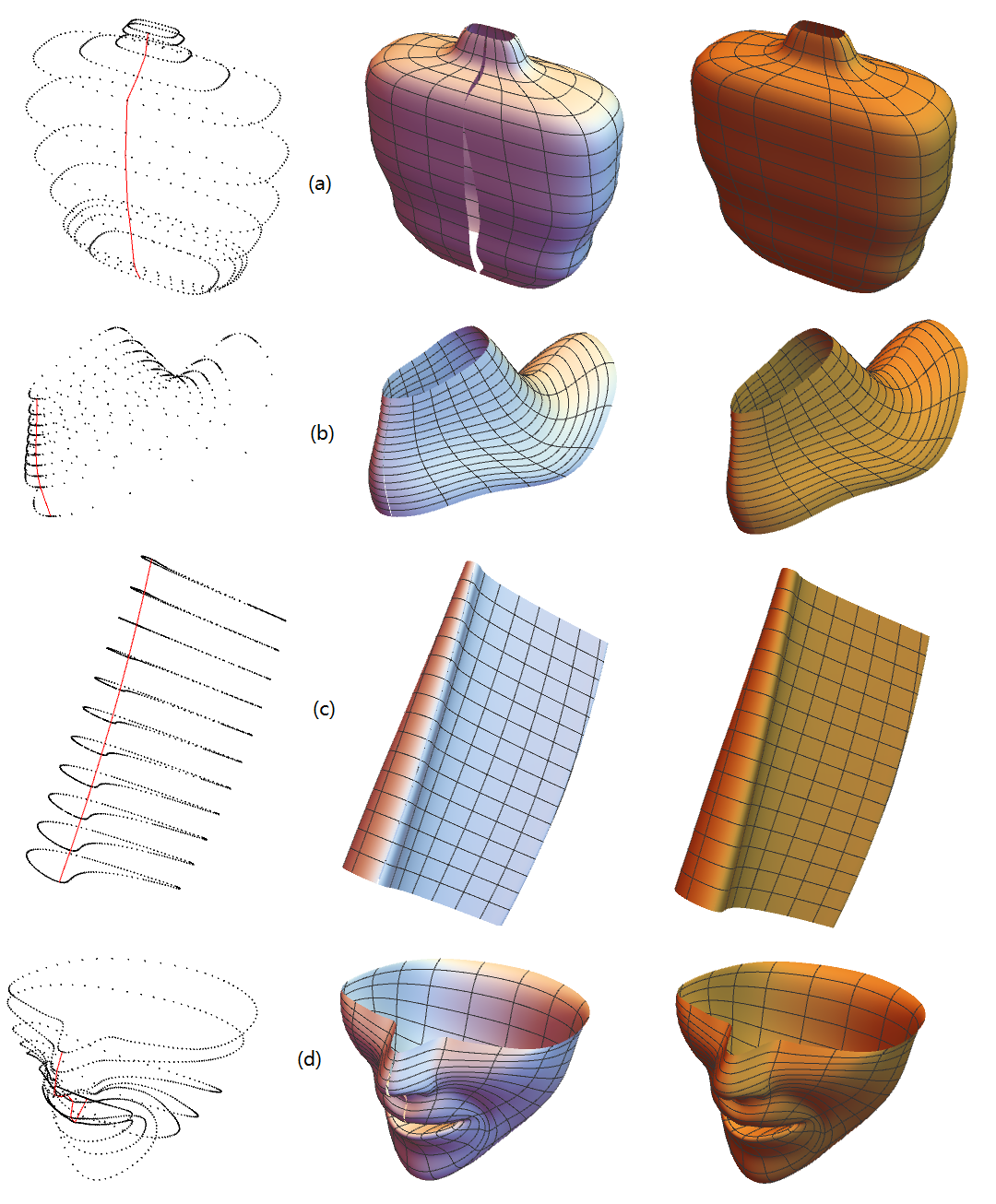}
    \caption{Lofted closed B-spline surfaces by the present method.
             Left: data points and the start-point polyline (red); middle: open visualization with visible seam; right: closed surface. (a)~Water bottle ($q=2$);
             (b)~shoe ($q=3$); (c)~blade ($q=4$); (d)~jaw ($q=5$).}
    \label{fig:surfaces}
\end{figure}

\begin{figure}[htbp]
    \centering
    \subfigure[Sampled points and start-point polyline.]{\includegraphics[height=3.1cm]{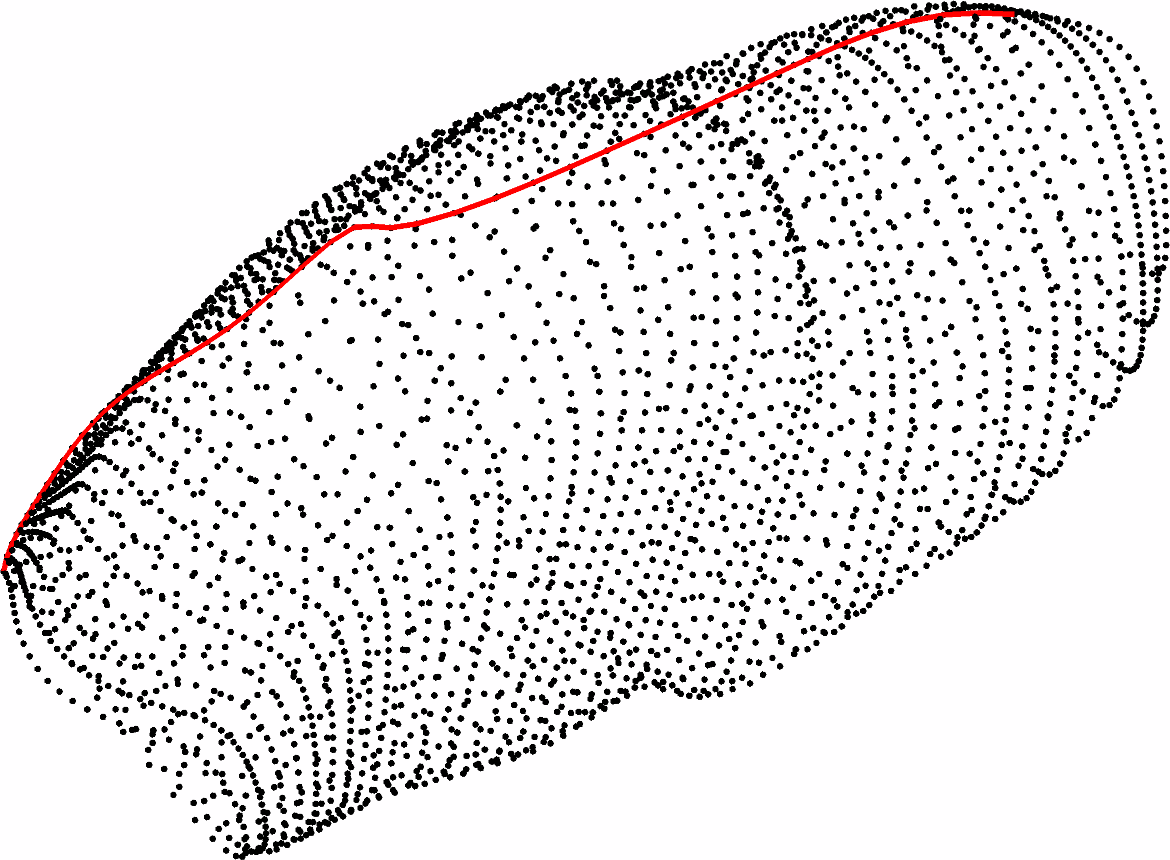}}
    \hspace{1mm}
    \subfigure[Control net ($48 \times 99$).]{\includegraphics[height=3.1cm]{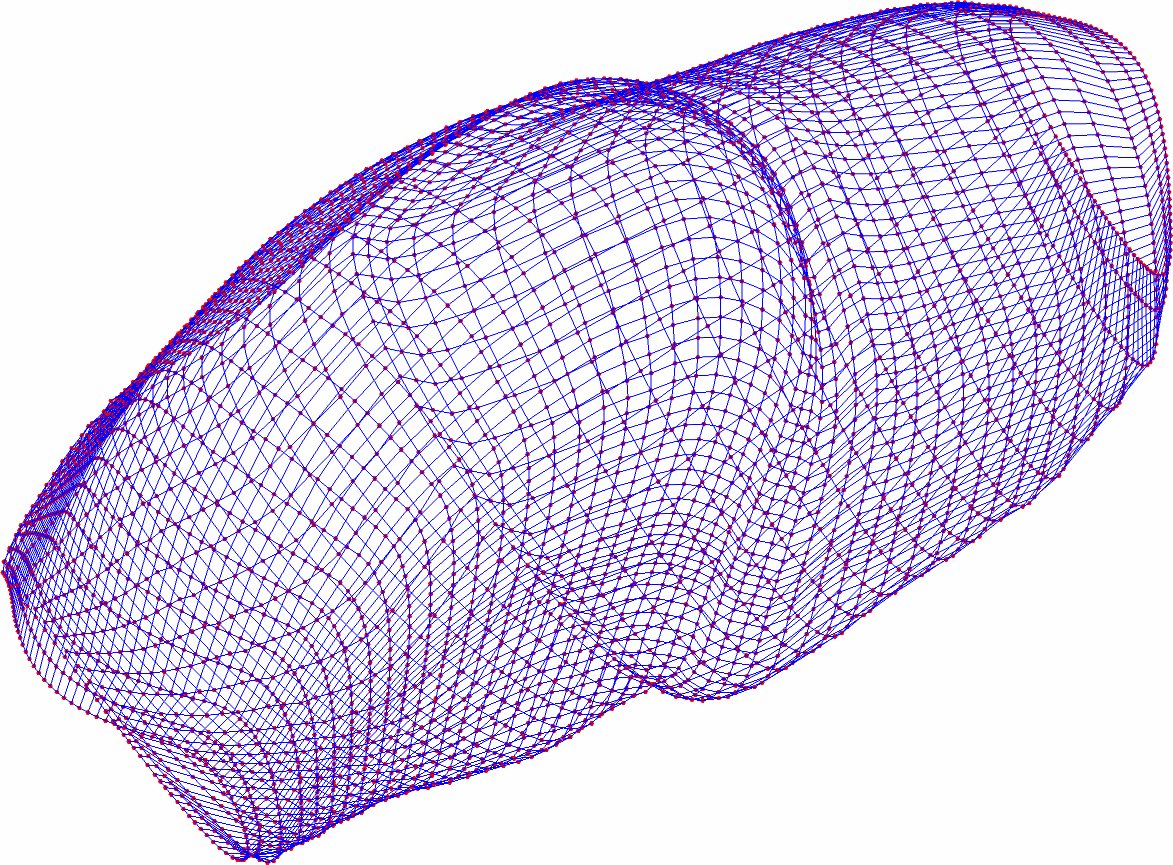}}
    \hspace{1mm}
    \subfigure[Second view.]{\includegraphics[height=3.2cm]{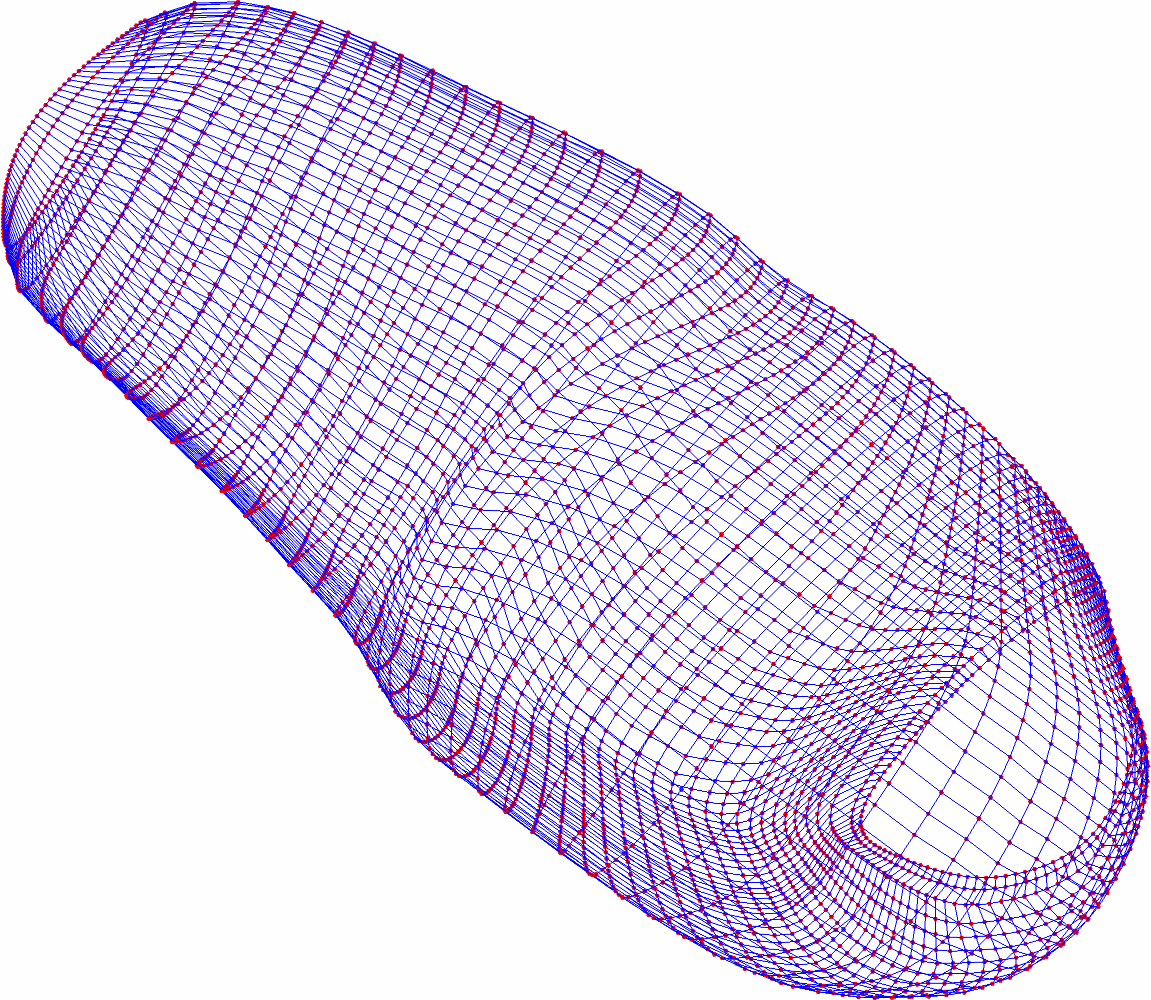}}
    \\[1mm]
    \subfigure[Lofted B-spline surfaces. Panels (a) and (c) of the strip show the open visualization from two viewpoints; panels (b) and (d) show the closed surface from the same viewpoints.]{\includegraphics[width=0.93\textwidth]{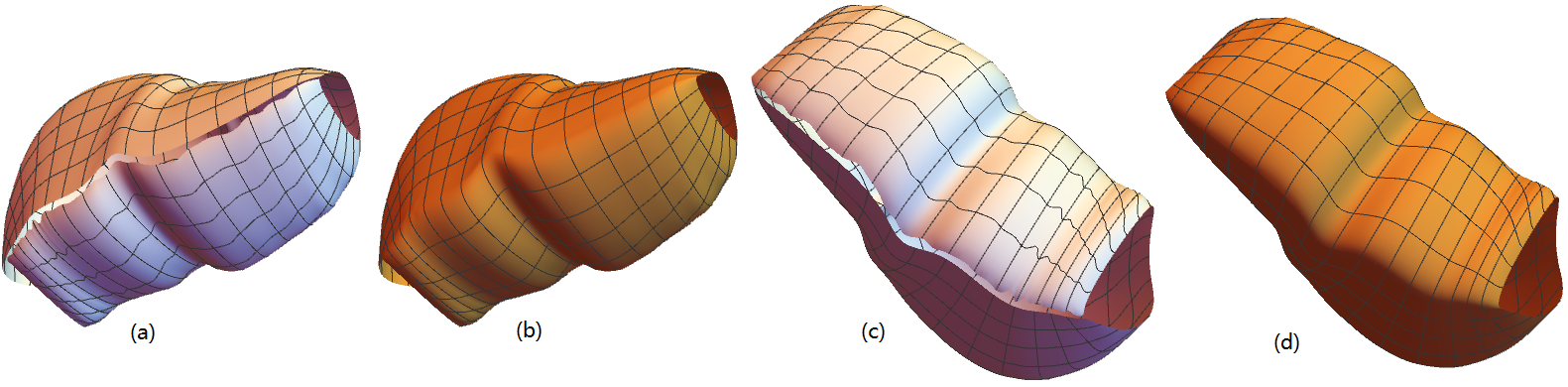}}
    \caption{Liver model via the present method with $(p,q)=(3,3)$, $\per=1.0$.}
    \label{fig:liver}
\end{figure}

\begin{figure}[htbp]
    \centering
    \subfigure[Sampled points and start-point polyline.]{\includegraphics[height=4.4cm]{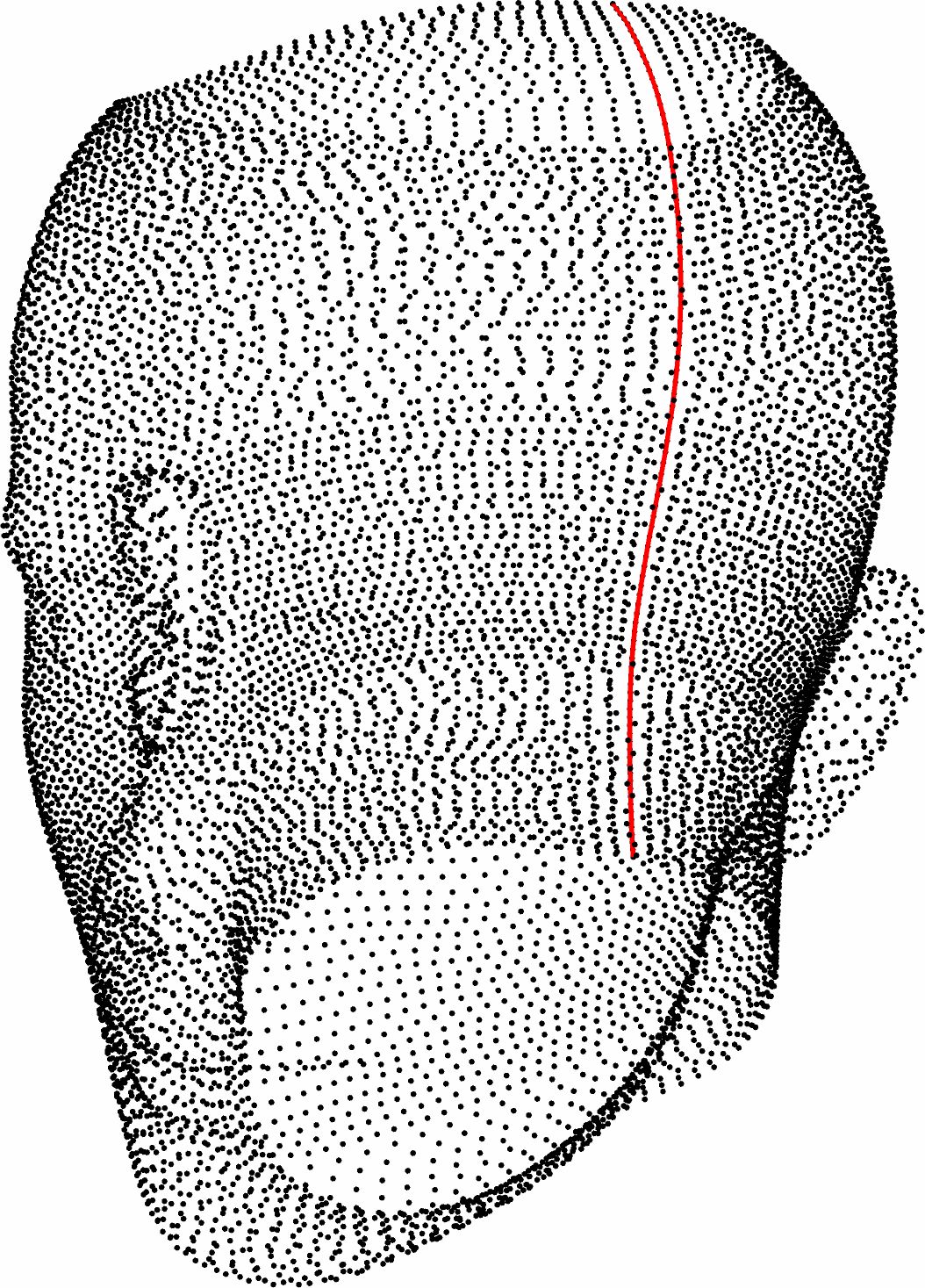}}
    \hspace{2mm}
    \subfigure[Control net ($92 \times 182$).]{\includegraphics[height=4.4cm]{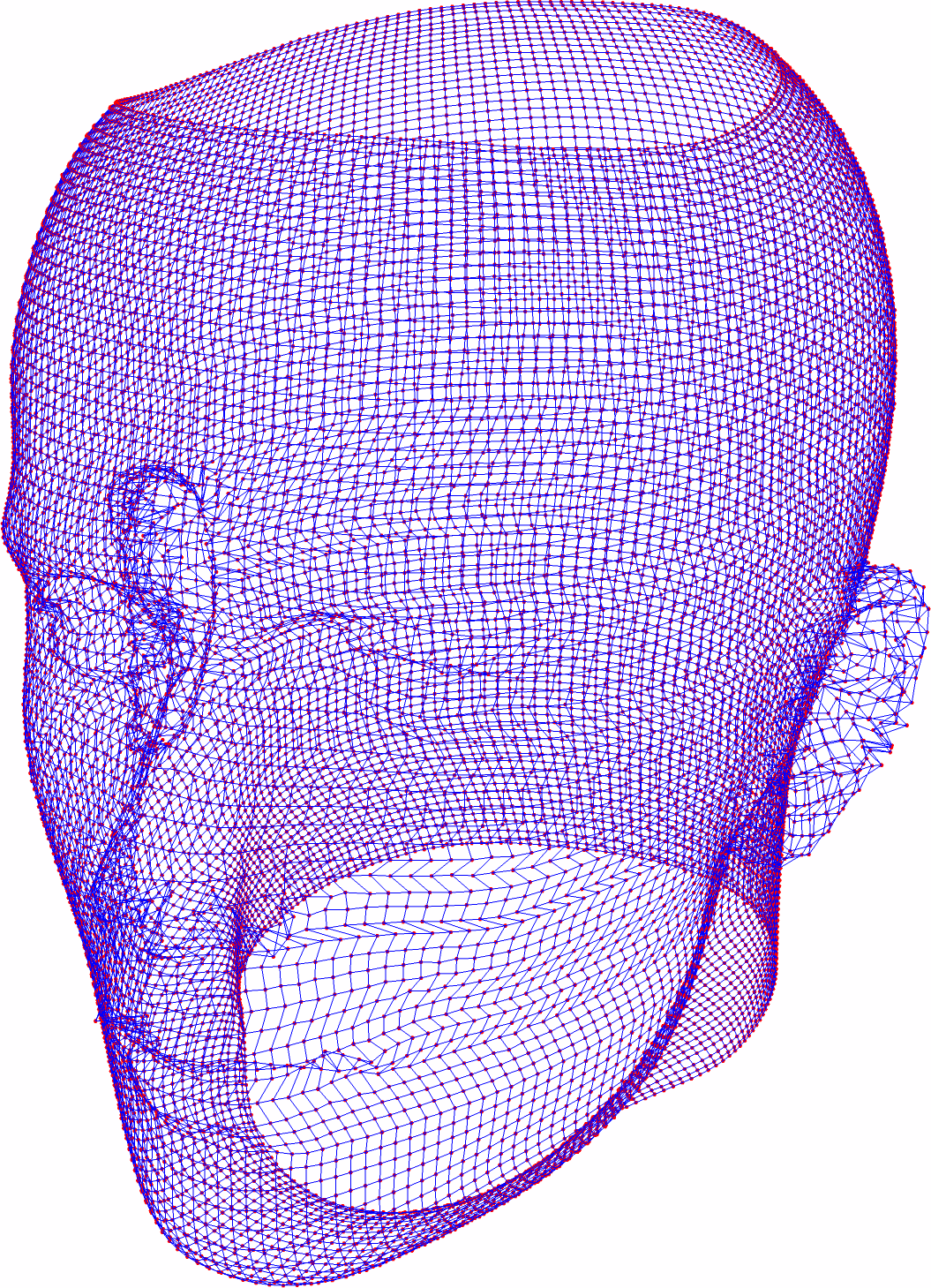}}
    \hspace{2mm}
    \subfigure[Second view.]{\includegraphics[height=4.1cm]{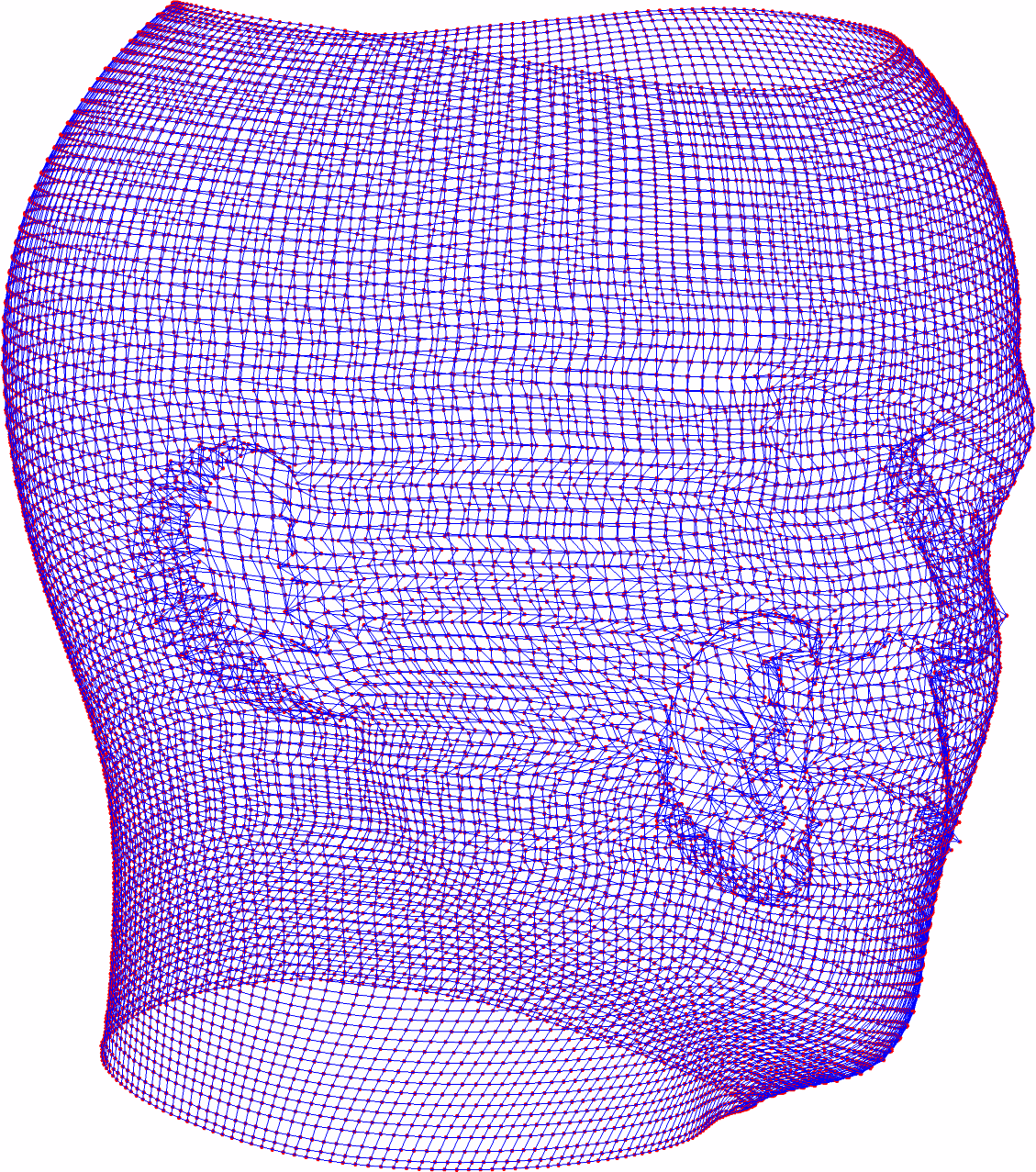}}
    \\[1mm]
    \subfigure[Lofted B-spline surfaces. Panels (a) and (c) of the strip show the open visualization from two viewpoints; panels (b) and (d) show the closed surface from the same viewpoints.]{\includegraphics[width=0.93\textwidth]{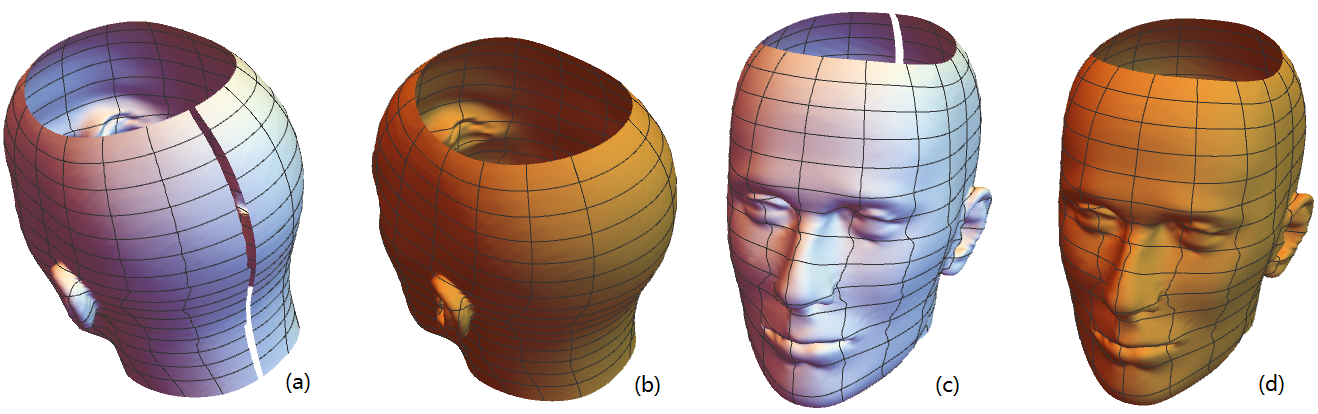}}
    \caption{Human head via the present method with $(p,q)=(3,3)$, $\per=1.0$.}
    \label{fig:head}
\end{figure}

\begin{table}[htbp]
    \centering
    \caption{Number of sections with odd and even point counts ($n_i+1$) for each model.}
    \label{tab:parity}
    \begin{tabular}{lccc}
        \toprule
        Model  & sections & odd $n_i+1$ & even $n_i+1$ \\
        \midrule
        Bottle & 12       & 7           & 5            \\
        Shoe   & 10       & 8           & 2            \\
        Blade  & 11       & 5           & 6            \\
        Jaw    & 9        & 5           & 4            \\
        Liver  & 48       & 21          & 27           \\
        Head   & 92       & 44          & 48           \\
        \midrule
        Total  & 182      & 90          & 92           \\
        \bottomrule
    \end{tabular}
\end{table}

\subsection{Numerical rank verification and cost}
\label{subsec:margins}

Table~\ref{tab:margins} gives the smallest rank-test margin over the sections of each model.
Only the sections with an even point count require the rank test, since odd point counts are certified by the window condition (Theorem~\ref{thm:odd-window}); 
among these, all selected square subsystems pass the test, and the anchor fallback of Procedure~\ref{lst:proc1} is never triggered; 
the same holds for the knot-selection baseline of~\cite{tang2023control}.
The fallback branch was therefore exercised separately on the uniform cubic counterexample of Proposition~\ref{prop:singular} and on a $10^{-10}$ knot perturbation.
Their computed margins, $8.3\times10^{-17}$ and $1.0\times10^{-10}$, are both rejected, whereas the anchor fallback configurations are accepted, 
with margins of order $10^{-4}$--$10^{-1}$ (Table~\ref{tab:margins}).
As reported in Table~\ref{tab:timing}, the full pipeline takes at most $0.37$~s, of which the per-section rank checks account for $0.02$--$0.20$~s.
All timings are wall-clock measurements recorded under a single-user load on an AMD Threadripper 9980X workstation.

The rank test depends on a single numerical tolerance $\sigma_{\mathrm{tol}}$, and its decisions are insensitive to this choice.
Every accepted candidate has a minimum singular value of at least $3.83\times10^{-4}$ (Table~\ref{tab:margins}), 
whereas the constructed singular and near-singular configurations have margins $8.3\times10^{-17}$ and $1.0\times10^{-10}$, respectively, i.e., at or below the floating-point noise level.
The accepted and rejected cases are thus separated by more than six orders of magnitude, with $\sigma_{\mathrm{tol}}=10^{-8}$ inside the gap; 
any tolerance in this gap yields the same accept/reject decisions on the six models.
The tolerance is therefore a round-off safeguard, not a parameter that tunes the outcome.

\begin{table}[htbp]
    \centering
    \caption{Per-section rank-test margins ($\per = 1.0$): minimum over the sections of each model of $\sigma_{\min}$ for the selected square subsystem and for the anchor configuration, at degrees $q = 2, 3, 4, 5$.}
    \label{tab:margins}
    \small
    \begin{tabular}{llcccc}
        \toprule
        Model & margin & $q{=}2$ & $q{=}3$ & $q{=}4$ & $q{=}5$\\
        \midrule
        \multirow{2}{*}{Bottle} & selected & $4.06\times10^{-2}$ & $1.22\times10^{-2}$ & $1.5\times10^{-2}$ & $5.63\times10^{-3}$\\
                                & anchors  & $4.55\times10^{-2}$ & $3.18\times10^{-2}$ & $1.48\times10^{-2}$ & $9.06\times10^{-3}$\\
        \multirow{2}{*}{Shoe}   & selected & $5.23\times10^{-2}$ & $5.14\times10^{-3}$ & $7.35\times10^{-4}$ & $3.83\times10^{-4}$\\
                                & anchors  & $3.12\times10^{-1}$ & $1.56\times10^{-1}$ & $8.48\times10^{-3}$ & $1.51\times10^{-3}$\\
        \multirow{2}{*}{Blade}  & selected & $9.32\times10^{-2}$ & $3.5\times10^{-2}$ & $2.18\times10^{-2}$ & $5.73\times10^{-3}$\\
                                & anchors  & $2.39\times10^{-1}$ & $1.53\times10^{-1}$ & $3.7\times10^{-2}$ & $1.82\times10^{-2}$\\
        \multirow{2}{*}{Jaw}    & selected & $9.17\times10^{-2}$ & $1.08\times10^{-1}$ & $2.12\times10^{-2}$ & $2.41\times10^{-2}$\\
                                & anchors  & $3.84\times10^{-1}$ & $2.48\times10^{-1}$ & $1.02\times10^{-1}$ & $5.38\times10^{-2}$\\
        \multirow{2}{*}{Liver}  & selected & $2.15\times10^{-1}$ & $1.3\times10^{-1}$ & $8.41\times10^{-2}$ & $3.52\times10^{-2}$\\
                                & anchors  & $4.83\times10^{-1}$ & $3.18\times10^{-1}$ & $1.7\times10^{-1}$ & $9.76\times10^{-2}$\\
        \multirow{2}{*}{Head}   & selected & $1.39\times10^{-1}$ & $9.23\times10^{-2}$ & $6.7\times10^{-2}$ & $4.53\times10^{-2}$\\
                                & anchors  & $4.93\times10^{-1}$ & $3.3\times10^{-1}$ & $2.03\times10^{-1}$ & $1.3\times10^{-1}$\\
        \bottomrule
    \end{tabular}
\end{table}

\begin{table}[htbp]
    \centering
    \caption{Timing components of the proposed method ($\per = 1.0$): full pipeline and the separately measured per-section rank checks.}
    \label{tab:timing}
    \begin{tabular}{lcccccc}
        \toprule
        & Bottle & Shoe & Blade & Jaw & Liver & Head\\
        \midrule
        full pipeline [s]  & 0.03 & 0.03 & 0.06 & 0.04 & 0.10 & 0.37\\
        rank checks [s]  & 0.02 & 0.02 & 0.04 & 0.02 & 0.05 & 0.20\\
        \bottomrule
    \end{tabular}
\end{table}

\subsection{Control-point counts}
\label{subsec:counts}

On these six cases at $\per=1$, proposed method uses 83--98\% fewer control points than traditional lofting and no more than the knot-selection method of~\cite{tang2023control} (Table~\ref{tab:counts-full}).
Because every selected subsystem passes, the safeguard introduces no fallback knots in these runs.
Table~\ref{tab:head-compare} examines scaling on the head model: the sub-models formed by the first $m + 1$ sections are lofted for increasing $m$ and window factors $\per \in \{0.25, 0.5, 0.75, 1.0\}$.
The two methods produce identical counts except at $\per = 1.0$ on the two largest sub-models, where the present method uses two fewer control points per section.

\begin{table}[htbp]
    \centering
    \caption{Control net sizes for the full models ($\per = 1.0$ and $q \in \{2, 3, 4, 5\}$).}
    \label{tab:counts-full}
    \begin{tabular}{lccccccc}
        \toprule
        Model & sections & $q$ & $\min(n_i)$ & $\max(n_i)$ & Traditional & Knot selection~\cite{tang2023control} & Proposed\\
        \midrule
        Bottle & 12 & 2 &  53 &  79 & $12 \times 801$   & $12 \times 83$  & $12 \times 83$\\
        Shoe   & 10 & 3 &  68 & 100 & $10 \times 926$   & $10 \times 144$ & $10 \times 135$\\
        Blade  & 11 & 4 & 139 & 159 & $11 \times 1601$  & $11 \times 221$ & $11 \times 221$\\
        Jaw    &  9 & 5 & 110 & 135 & $9 \times 1114$   & $9 \times 195$  & $9 \times 193$\\
        Liver  & 48 & 3 &  39 &  89 & $48 \times 3035$  & $48 \times 99$  & $48 \times 99$\\
        Head   & 92 & 3 &  89 & 178 & $92 \times 10630$ & $92 \times 184$ & $92 \times 182$\\
        \bottomrule
    \end{tabular}
\end{table}

\begin{table}[htbp]
    \centering
    \caption{Control net row sizes for the human head model by different $m$ and $\per$ ($q = 3$).}
    \label{tab:head-compare}
    \begin{tabular}{cccccccc}
        \toprule
        $m$                              &             &   5 &  10 &  20 &  40 &  80 &  91 \\ \cline{1-8}
        $\min(n_i)$                      &             &  90 &  90 &  90 &  90 &  89 &  89 \\
        $\max(n_i)$                      &             &  97 &  97 &  99 & 178 & 178 & 178 \\
        $\mathrm{avg}(n_i)$              &             &  93 &  93 &  94 & 123 & 116 & 114 \\ \cline{1-8}
        Traditional (closed)~\cite{tang2023control} & {}  & 656 & 1123 & 2086 & 5237 & 9590 & 10630 \\ \cline{1-8}
        \multirow{4}{*}{Knot selection (closed)~\cite{tang2023control}} & $\per=1.00$ & 101 &  101 &  103 &  182 &  184 &  184 \\
                                          & $\per=0.75$ & 131 &  161 &  181 &  346 &  353 &  353 \\
                                          & $\per=0.50$ & 169 &  211 &  230 &  423 &  434 &  434 \\
                                          & $\per=0.25$ & 220 &  307 &  383 &  772 &  835 &  835 \\ \cline{1-8}
        \multirow{4}{*}{Proposed}         & $\per=1.00$ & 101 &  101 &  103 &  182 &  182 &  182 \\
                                          & $\per=0.75$ & 131 &  161 &  181 &  346 &  353 &  353 \\
                                          & $\per=0.50$ & 169 &  211 &  230 &  423 &  434 &  434 \\
                                          & $\per=0.25$ & 220 &  307 &  383 &  772 &  835 &  835 \\
        \bottomrule
    \end{tabular}
\end{table}

\subsection{Energy comparison}
\label{subsec:fairness}

The proposed method minimizes~\eqref{eq:energy} over the interpolants in its own common-knot spline space.
Let $E_{\mathrm{sel}}$ and $E_{\mathrm{prop}}$ denote the section strain energies of the previous knot-selection method~\cite{tang2023control} and of the proposed method, respectively.
So Table~\ref{tab:energy} reports a direct comparison using identical section data, parametrization, degrees, energy weights, and $\per=1.0$.
Across all 182 sections, the proposed method has no larger computed strain energy within numerical precision.
The largest ratio $E_{\mathrm{sel}}/E_{\mathrm{prop}}$ is $1.094$; equivalently, the proposed method's energy is up to $8.6\%$ lower.

\begin{table}[htbp]
    \centering
    \caption{Section strain energy for the knot-selection method~\cite{tang2023control} and the proposed method ($\per=1.0$); ratios are $E_{\mathrm{sel}}/E_{\mathrm{prop}}$ per section.}
    \label{tab:energy}
    \begin{tabular}{lccccc}
        \toprule
        Model & $q$ & mean $E_{\mathrm{sel}}$ & mean $E_{\mathrm{prop}}$ & mean ratio & max ratio \\
        \midrule
        Bottle       & 2 & $6.47 \times 10^{1}$ & $6.42 \times 10^{1}$ & 1.006 & 1.011 \\
        Shoe         & 3 & $4.69 \times 10^{4}$ & $4.68 \times 10^{4}$ & 1.003 & 1.006 \\
        Blade        & 4 & $2.64 \times 10^{5}$ & $2.63 \times 10^{5}$ & 1.005 & 1.030 \\
        Jaw          & 5 & $9.12 \times 10^{4}$ & $9.11 \times 10^{4}$ & 1.002 & 1.013 \\
        Liver        & 3 & $2.74 \times 10^{3}$ & $2.74 \times 10^{3}$ & 1.002 & 1.007 \\
        Head         & 3 & $1.15 \times 10^{3}$ & $1.13 \times 10^{3}$ & 1.015 & 1.094 \\
        \bottomrule
    \end{tabular}
\end{table}

\section{Discussion and analysis}
\label{sec:discussion}

\subsection{Even-cardinality behavior and the anchor fallback}
The exact singularity condition at even cardinality remains open, so the algorithm applies a numerical rank test and falls back to anchor knots when the selected first-fit candidate fails.
The full-row-rank statement now proved as Theorem~\ref{thm:full-rank-lift} originated in an earlier manuscript from 2018, 
later posted as arXiv:2202.06330~\cite{tang2022optimal}, and the first-fit construction used here was developed from that statement.
At that time the parity dependence of periodic B-spline interpolation was not recognized; the present paper isolates it through Proposition~\ref{prop:singular}.
The fallback is therefore a safeguard rather than the normal operating path.
In the experiments reported here the selected first-fit candidates pass the rank test on every even-cardinality section, 
and the odd-cardinality sections are certified by the window condition, so the fallback is never triggered on the 182 sections (Table~\ref{tab:margins}); the same holds for the knot-selection baseline of Tang~\cite{tang2023control} using Procedure~1 of that paper.
The same algorithm has also been used extensively by the authors and their industrial collaborators on closed-contour data with chord-length parameterization; 
in our experience it has not exhibited numerical instability, and the resulting surfaces were visually pleasant.

We therefore conjecture that, under chord-length or centripetal parameterization, the probability of obtaining a singular collocation matrix at even cardinality is very small.
The counterexample of Proposition~\ref{prop:singular} requires a coordinated placement of parameters and knots, 
with parameters at the midpoints between consecutive knots for odd degree, and such a critical configuration is unlikely to arise from measured contours.
A structural characterization that quantifies this likelihood, or that proves the first-fit candidate nonsingular under a checkable parameterization condition, remains open.

\subsection{Choice of the window factor \texorpdfstring{$\per$}{per}}
\label{subsec:per-tradeoff}

The window factor balances knot sharing against geometric freedom.
At $\per=0$ the selection collapses to the anchor knots, and we use the per-section anchor interpolation as the no-sharing reference; 
for $\per>0$ the common-knot energy minimizer is used. At $\per=1$ the selection uses the full flexibility windows and maximizes sharing.
Figure~\ref{fig:per-tradeoff}(a) reports the resulting number of control points per section for the six models as $\per$ varies over $\{0,0.2,\ldots,1.0\}$.
The count decreases monotonically with $\per$ in every model.
At $\per=0.8$ the method already achieves $94$--$98\%$ of the total control-point reduction obtained at $\per=1$, 
and the final step from $0.8$ to $1.0$ still removes a substantial number of control points on the bottle, liver, and head models.

Figure~\ref{fig:per-tradeoff}(b) reports the mean section strain energy normalized by its value at $\per=0$.
For five of the six models the relative energy remains within about $1.4\%$ over the whole range, 
and the variation is not monotone: the blade, bottle, and jaw have slightly lower energy at intermediate values, whereas the shoe and liver change by less than $0.5\%$.
The head is the exception: its relative energy stays near $1.006$ up to $\per=0.8$ and then rises to $1.087$ at $\per=1$, while the control count only drops from $340$ to $182$.
Thus the last increment of sharing has a small compression return and a comparatively large fairness cost for this model.

The practical guideline is model-dependent.
When compression is the primary objective, $\per=1.0$ is appropriate.
When fairness is prioritized, or when the energy curve rises sharply near $\per=1$, $\per\approx0.8$ is a better compromise: 
it retains most of the compression and keeps the energy close to the no-sharing endpoint.
For the head model, for example, $\per=0.8$ uses $340$ rather than $182$ control points but keeps the relative energy near $1.006$ instead of $1.087$.
Because the trade-off is not uniform across models, $\per$ should be chosen per model rather than fixed globally.

\begin{figure}[htbp]
    \centering
    \subfigure[Control points per section (log scale).]{\includegraphics[width=0.48\textwidth]{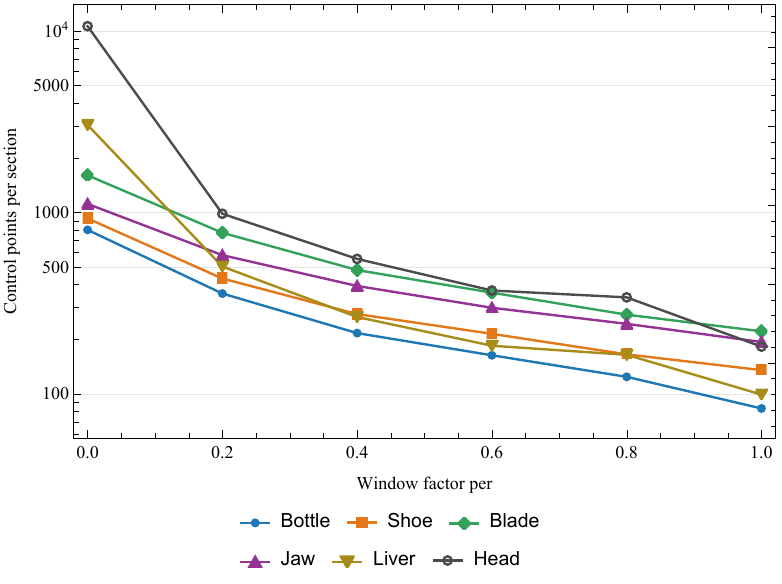}}
    \hfill
    \subfigure[Mean section strain energy relative to $\per=0$.]{\includegraphics[width=0.48\textwidth]{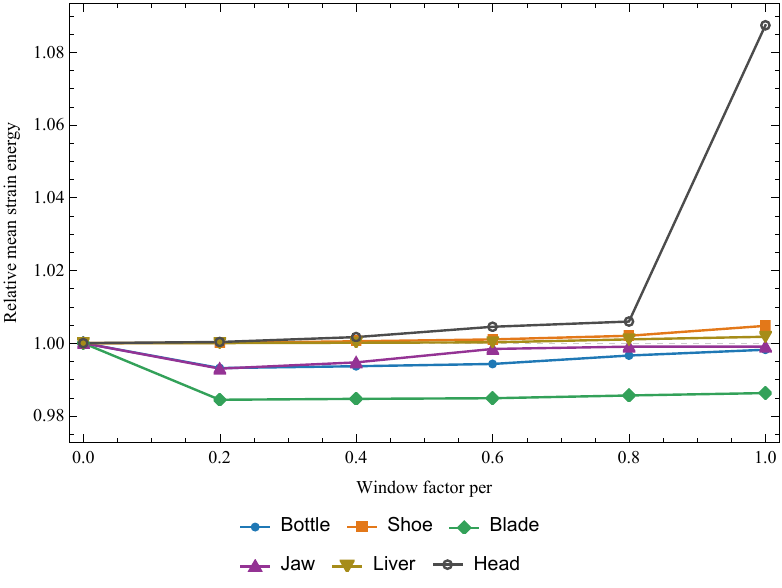}}
    \caption{Effect of the window factor $\per$ on the six models: (a) control points per section on a logarithmic scale; (b) mean section strain energy normalized by its value at $\per=0$.}
    \label{fig:per-tradeoff}
\end{figure}

\subsection{Comparison with T-spline skinning}
The closest T-spline work to ours is the periodic T-spline skinning of Li~\etal\cite{li2015surface}, 
which operates on closed NURBS section curves and represents the lofted T-spline surface in semi-NURBS form.
Our input is different: we start from rows of closed contour points rather than from compatible NURBS curves.
Applying the periodic T-spline construction to point rows therefore requires an interpolation step for every row.
If each row is interpolated independently, the row knots again accumulate in the tensor-product (or T-spline) representation, 
and the control-point explosion that motivates this paper reappears before T-spline skinning can help.
If, instead, the rows are interpolated with the knot-reuse construction of Tang~\cite{tang2023control}, 
the input curves already share a compact knot set; 
T-spline skinning then naturally uses fewer control points than tensor-product B-spline lofting.
Such a pipeline would compare a T-spline representation built on an already reduced input with a tensor-product method, 
and would not isolate the effect of the lofting representation itself; a direct quantitative comparison would therefore not be meaningful.

Li~\etal also compare under this favorable setting: in their examples the input NURBS curves have relatively few control points and share many knots, 
as the figures in~\cite{li2015surface} illustrate, and they report a reduction by a factor of roughly 1.5--2 relative to NURBS lofting.
Our experiments address a different question, namely how many control points are needed to interpolate closed contour point rows with certified solvability, 
and report reductions of 83--98\% relative to traditional interpolating lofting.
The two results are therefore complementary rather than directly comparable: T-spline skinning provides local refinement and a more flexible representation for compatible curve inputs, 
whereas the present method supplies the missing solvability foundation and knot reduction for tensor-product interpolation of closed point rows.

\section{Conclusions and future work}
\label{sec:conclusions}

Full row rank of a section collocation matrix guarantees interpolation for closed section data and invertibility of the corresponding energy-minimizing saddle system.
The periodic window condition previously used to infer this rank is sufficient at odd point counts (Theorem~\ref{thm:odd-window}) but may fail at even point counts, 
as the uniform and non-uniform counterexamples of Proposition~\ref{prop:singular} show.
The present method therefore combines the natural and shifting anchor knots, 
which are nonsingular on given parameter values, with a parity-aware rank check: at odd point counts the first-fit candidate is accepted by the window condition, 
and only at even point counts is a singular-value test applied, with the anchor configuration as a safe fallback.
The full-row-rank lifting theorem (Theorem~\ref{thm:full-rank-lift}) transfers this square subsystem to the rectangular common-knot system, 
making the energy-minimizing construction feasible for closed section data.

On six case-study models at degrees $q=2,3,4,5$ and window factor $\per=1$, 
the method used 83--98\% fewer control points than traditional closed lofting and no more than the previous knot-selection method, with no larger section energies.
The rank decisions are insensitive to the numerical tolerance, because the accepted and rejected margins are separated by more than six orders of magnitude.
The main remaining problem is the even-cardinality case: 
characterizing the exact singularity condition of the periodic collocation matrix at even point counts would replace the present numerical rank check by a structural singular/non-singular test, 
making the construction fully parity-independent.

\appendix

\section{Proofs}
\label{app:proofs}

\begin{proof}[Proof of Lemma~\ref{lem:nested}]
    Let ${\bf{U}}_\infty$ be the bi-infinite periodization of ${\bf{D}}$ and $N^{\infty}_{j}$, $j \in \mathbb{Z}$, the B-splines on it, indexed as in Section~\ref{subsec:periodic}; since $\hat n \ge q$ each support spans at most one period, and the periodic extension of $N^{\mathrm{per}}_{i,q,{\bf{D}}}$ equals $\sum_{\ell} N^{\infty}_{i+\ell(\hat n+1)}$, which is the periodization~\eqref{eq:per-basis}.

    \emph{Local independence.} On any open interval, a vanishing combination of the periodic basis lifts to a locally finite combination of the ordinary B-splines that meet that interval.
    Their local linear independence~\cite{de1978practical} forces the coefficients of every active residue class to vanish; hence the periodic basis is locally, and therefore globally, linearly independent.

    \emph{Span.} Let $\mathcal{P}$ be the space of $1$-periodic $C^{q-1}$ piecewise polynomials of degree $\le q$ with breakpoints in ${\bf{U}}_\infty$; $\dim \mathcal{P} = \hat n + 1$~\cite[\S 8.1]{schumaker2007spline}.
    Each $f \in \mathcal{P}$ has a unique Curry--Schoenberg representation $f = \sum_j a_j N^{\infty}_{j}$ on the line~\cite{de1978practical}; periodicity of $f$ and uniqueness force $a_{j+(\hat n+1)} = a_j$, and regrouping the sum modulo $\hat n+1$ exhibits $f \in \mathbb{S}_{q,{\bf{D}}}$.
    With independence, $\mathbb{S}_{q,{\bf{D}}} = \mathcal{P}$.

    \emph{Nesting.} Any $f \in \mathbb{S}_{q,\bar{{\bf{D}}}}$ is $1$-periodic, $C^{q-1}$, and polynomial on each ${\bf{D}}$-span (each is contained in a $\bar{{\bf{D}}}$-span), so $f \in \mathbb{S}_{q,{\bf{D}}}$; expanding the coarse basis in the fine one defines ${\bf{M}}$.
    If ${\bf{c}}^{\mathsf{T}}{\bf{M}} = {\bf{0}}$ then $\sum_i c_i N^{\mathrm{per}}_{i,q,\bar{{\bf{D}}}} \equiv 0$, so ${\bf{c}} = {\bf{0}}$ by independence on $\bar{{\bf{D}}}$ (using $n \ge q$); hence $\operatorname{rank} {\bf{M}} = n+1$.
\end{proof}


\bibliographystyle{unsrtnat}
\bibliography{refs}

\end{document}